\documentclass[]{article_custom}

\def \PP{\mathbb P}
\def \RR{\mathbb R}

\def \UU{\mathbb U}

\def \cD{\bm \Theta}
\def \cP{\mathcal D}

\def \cU{\mathcal U}
\def \cV{\mathcal V}

\def \cO{\mathcal O}
\def \cB{\mathcal B}

\def \FH{Fr\'echet-Hoeffding }

\def \Q{G^{-1}}
\def \Qa{\Q(\alpha)}
\def \Qastar{{G ^{-1}}^\star(\alpha)}
\def \Qt{G_{\btheta}^{-1}}
\def \Qta{\Qt(\alpha)}

\def \Qth{\widehat G_{\btheta}^{-1}}
\def \Qtha{\Qth(\alpha)}

\def \bX{\mathbf X}
\def \bB{\mathbf B}

\def \bx{\textbf x}

\def \bu{\mathbf u}

\def \btheta{\bm \theta}

\def \bbeta{\bm \beta}

\def \bTheta{\bm \Theta}

\def \Ptheta {P_{\btheta}}

\def \Ctheta{C_{\btheta}}
\def \ctheta{c_{\btheta}}
\def \Ftheta{F_{\btheta}}
\def \ftheta{f_{\btheta}}
\def \Gtheta{G_{\btheta}}
\def \gtheta{g_{\btheta}}
\def \Qtheta{G_{\btheta}^{-1}}

\def \Gthetan{\widehat G_{\btheta}}

\def \htheta{\hat \btheta}
\def \hthetaN{\widehat \btheta_{N}}

\def \Gthetah{G_{\btheta + h}}
\def \fthetah{f_{\btheta + h}}
\def \Fthetah{F_{\btheta + h}}

\newcommand{\argmin}{\operatornamewithlimits{argmin}}

\usepackage[T1]{fontenc}
\usepackage[utf8]{inputenc}
\usepackage{lmodern}
\usepackage{graphicx} % Graphic package
\usepackage{framed} % Frame algorithm, tables, etc...
\usepackage[shortlabels]{enumitem} % Enumeration lists

\usepackage[numbers,sort&compress]{natbib}% Citation support using natbib.sty
\bibpunct[, ]{[}{]}{,}{n}{,}{,}% Citation support using natbib.sty
\makeatletter% @ becomes a letter
\def\NAT@def@citea{\def\@citea{\NAT@separator}}% Suppress spaces between citations using natbib.sty
\makeatother% @ becomes a symbol again

\usepackage{bm} % To bold the grec symbols
\usepackage{bbm} % To bold the grec symbols
\usepackage{bbold} % Kind of extension of mathbb
\usepackage[caption=false]{subfig}
\usepackage{grffile}
\usepackage{ulem}
\usepackage{xcolor}
\usepackage{filecontents}

\usepackage{proof}

\newtheorem{mydef}{Definition}
\newtheorem{theor}{Theorem}

\newtheorem{prop}{Proposition}
\newtheorem{assumption}{Assumption}

\renewcommand \theassumption{$\mathbf{\Alph{assumption}}$}
\let \origtheassumption \theassumption

\usepackage[vlined, ruled, linesnumbered]{algorithm2e}
\usepackage{algpseudocode}

\SetKwProg{Fn}{function}{:}{}
\SetKwRepeat{Do}{do}{while}

\SetAlFnt{\sffamily}

\SetAlCapFnt{\normalfont\sffamily\large}

\SetAlCapSkip{.5em}

\SetCommentSty{mycommfont}

\makeatletter
\renewcommand{\algocf@Vline}[1]{%     no vskip in between boxes but a strut to separate them, 
  \strut\par\nointerlineskip% then interblock space stay the same whatever is inside it
  \algocf@push{\skiprule}%        move to the right before the vertical rule
  \hbox{\bgroup\color{cyan}\vrule\egroup%
    \vtop{\algocf@push{\skiptext}%move the right after the rule
      \vtop{\algocf@addskiptotal #1}\bgroup\color{cyan}\Hlne\egroup}}\vskip\skiphlne% inside the block
  \algocf@pop{\skiprule}%\algocf@subskiptotal% restore indentation
  \nointerlineskip}% no vskip after
\renewcommand{\algocf@Vsline}[1]{%    no vskip in between boxes but a strut to separate them, 
  \strut\par\nointerlineskip% then interblock space stay the same whatever is inside it
  \algocf@bblockcode%
  \algocf@push{\skiprule}%        move to the right before the vertical rule
  \hbox{\bgroup\color{cyan}\vrule\egroup%               the vertical rule
    \vtop{\algocf@push{\skiptext}%move the right after the rule
      \vtop{\algocf@addskiptotal #1}}}% inside the block
  \algocf@pop{\skiprule}% restore indentation
  \algocf@eblockcode%
}
\makeatother
\usepackage{authblk}

\title{Detecting and modeling worst-case dependence structures between random inputs of computational reliability models}

\author[1, 4]{\small Nazih Benoumechiara}
\author[2]{\small Bertrand Michel}
\author[3]{\small Philippe Saint-Pierre}
\author[1, 5]{\small Nicolas Bousquet}

\affil[1]{\footnotesize Sorbonne Universit\'e, Laboratoire de Probabilit\'es, Statistique et Mod\'elisation, 75005 Paris, France}
\affil[2]{\footnotesize LMJL-UMR6629, \'Ecole Centrale de Nantes, Nantes, France}
\affil[3]{\footnotesize Institut de Mathématiques de Toulouse, UMR 5219 Université de Toulouse UPS IMT, F-31062 Toulouse Cedex 9, France}
\affil[4]{\footnotesize EDF R\&D, Chatou, France}
\affil[4]{\footnotesize Quantmetry, Paris, France}

\date{}

\begin{document}

\maketitle

\begin{abstract}
Uncertain information on input parameters of reliability models is usually modeled by considering these parameters as random, and described by marginal distributions and a dependence structure of these variables. In numerous real-world applications, while information is mainly provided by marginal distributions, typically from samples, little is really known on the dependence structure itself. Faced with this problem of incomplete or missing information, risk studies are often conducted by considering independence of input variables, at the risk of including irrelevant situations. This approach is especially used when reliability functions are considered as black-box computational models. Such analyses remain weakened in absence of in-depth model exploration, at the possible price of a strong risk misestimation. Considering the frequent case where the reliability output is a quantile, this article provides a methodology to improve risk assessment, by exploring a set of pessimistic dependencies using a copula-based strategy. In dimension greater than two, a greedy algorithm is provided to build input regular vine copulas reaching a minimum quantile to which a reliability admissible limit value can be compared, by selecting pairwise components of sensitive influence on the result. The strategy is tested over toy models and a real industrial case-study. The results highlight that current approaches can provide non-conservative results, and that a nontrivial dependence structure can be exhibited to define a worst-case scenario.
\end{abstract}

\section{Introduction}

Many industrial companies, like energy producers or vehicle and aircraft manufacturers, have to ensure a high level of safety for their facilities or products. In each case, the structural reliability of certain so-called critical components plays an essential role in overall safety. For reasons related to the fact that these critical components are highly reliable, and that real robustness tests can be very expensive or even hardly feasible, structural reliability studies generally use simulation tools \cite{de2008uncertainty, lemaire2010}. The physical phenomenon of interest being reproduced by a numerical model $\eta$ (roughly speaking, a {\it computer code}), such studies are based on the calculation of a reliability indicator based on the comparison of $y=\eta(\bx)$ and a safety margin, where $\bx$ corresponds to a set of input parameters influencing the risk. In the framework of this article, such models are considered as black box and can be explored only by simulation means.

While the problems of checking the validity of $\eta$ and selecting inputs $\bx \in \chi \subseteq \RR^d$ are addressed by an increasing methodological corpus \cite{Bayarri2007, NRC2012}, a perennial issue is the modeling of $\bx$. Differing from the specification of $\eta$ itself, this input vector is  known with uncertainty,  either because the number of experiments to estimate is limited, or because some inputs reflect intrinsically variable phenomena \cite{NIL03}. In most cases, these epistemic and aleatory uncertainties are jointly modeled by probability distributions \cite{Helton2011}. Consecutively, the reliability indicator is often defined as the probability that $y$ be lower than a threshold ({\it failure probability}), or a limit quantile for $y$. This article focuses on this last indicator, which provides an upper or lower bound of the mean effect of the output variable uncertainty.

Therefore the modeling of $\bx$ stands on the assessment of a joint probability distribution with support $\chi$, divided between marginal and dependencies features. Though information on each dimension of $\bx$ can often be accessible experimentally or using physical or expert knowledge \cite{bedford2006}, the dependence structure between the component of $\bx$ remains generally unknown. Typically, statistical data are only available per dimension, but not available for two or more dimensions simultaneously. For this reason, most of robustness studies are conducted by sampling within independent marginal distributions. Doing so, reliability engineers try to capture input situations that minimize the reliability indicator. Such situations are defined as so-called {\it worst cases}. However, the assumption of independence between inputs has been severely criticized since the works by \cite{Grigoriu79} and \cite{Thoft82}, who showed that output failure probabilities of industrial systems can significantly vary and be underestimated if the input dependencies are neglected. More generally, \cite{Tang2013impactcopula, tang2015copula} showed that tail dependencies between inputs can have major expected effects on the uncertainty analysis results.

Returning to a probabilist framework, and beyond structural reliability, the problem of defining a worst-case scenario by selecting a joint input distribution, from incomplete information, is a topical issue encountered in many fields. In decision-making problems, \cite{scarf1958min} proposed a general definition of the worst case distribution as the minimizer of an excepted cost among a set of possible distributions. More recently, \cite{Agrawal12} extended this approach to account for incomplete dependence information. These theoretical works, that propose selection rules over the infinite set of all possible joint distributions, remain hard to apply in practice. Recent applied works made use of copulas \cite{Nelsen07} to model dependencies between stochastic inputs \cite{Tang2013impactcopula, tang2015copula}, following other researchers confronted to similar problems in various fields: finance \cite{Cherubini2004copula}, structural safety \cite{goda2010statistical}, environmental sciences \cite{schoelzel2008multivariate} or medicine \cite{beaudoin2008archimedean}.  These studies mainly consider bivariate copulas, which makes theses analysis effective only when two random variables are correlated. Cases where a greater number of variables is involved were explored by \cite{jiang2015vine}, who used \textit{vine copulas} to approach complex multidimensional correlation problems in structural reliability. A vine copula is a graphical representation of the \textit{pair-copula construction} (PCC), proposed by \cite{joe1996families}, which defines a multidimensional dependence structure using conditional bivariate copulas. Various class of vines exist (see \cite{czado2010pair} for a review), and among them the regular vines (R-vines) introduced by \cite{bedford2001probability, bedford2002vines} are known for their appealing computational properties, while inference on PCC is usually demanding \cite{dissmann2013selecting,Haff2016}. 

R-vine parametric copulas seem promising to improve the search for a worst-case dependence between stochastic inputs, while keeping the benefits of a small number of parameters, as favoring inference and conducting simple sensitivity analyses \textit{a posteriori}. To our knowledge, however, no practical methodology has been yet proposed to this end for which the notion of worst case is defined by the minimization of an output quantile. This is the subject of this article. More precisely, the aim of this research is to determine a parametric copula over $\bx$, close to the worst case dependence structure, which is associated to a minimum value of the quantile of the distribution of $y$. Given a vine structure defined by a parameter vector,  the  optimization problem involves to conduct empirical quantile estimations for each value of this vector in a finite set of interest (chosen as a grid). The proposed methodology stands on an encompassing greedy algorithm exploring copula structures, which integrates several sub-algorithms of increasing complexity and is based on some simplifying assumptions. These algorithms are made available in the Python library \texttt{dep-impact} \cite{depimpact}.

The article is therefore organized as follows. Section \ref{sec:consistency} introduces the framework and studies the consistency of a statistical estimation of the minimum quantile, given an input copula family and a growing sequence of grids. A preliminary study of the influence of the dependence structure, specific to quantile minimization, is conducted in Section \ref{sec:methodology} as a first application of this statistical optimization. The wider problem of selecting copulas in high-dimensional settings using a sequence of quantile minimization is considered in Section \ref{sec:quantile_minimization}. While the choice of R-vines is defended, a sparsity hypothesis is made to diminish the computational burden, according to which only a limited number of pairwise dependencies is influent on the result. A greedy algorithm is proposed to carry out the complete procedure of optimization and modeling. This heuristic is tested in Section \ref{sec:applications} over toy examples, using simulation, and a real industrial case-study. The results highlight that worst-case scenarios produced by this algorithm are often bivariate copulas reaching the \FH bounds \cite{Hoeffding40, Frechet51} (describing perfect dependence between variables), as it could be expected in monotonic frameworks, but that other nontrivial copulas can be exhibited in alternative situations. Results and avenues for future research are extensively discussed in the last section of this article. We also refer to Appendix \ref{sec:proof_consitency} and \ref{sec:vine_copulas} for supplementary material on consistency proofs, on R-vine copulas and on R-vine iterative construction.
\section{Minimization of the quantile of the output distribution}
\label{sec:consistency}

This section introduces a general framework for the calculation of the minimum quantile of the output distribution of a computational model, when the input distribution can be taken from a large family of distributions, each one corresponding to a particular choice of dependencies between the input variables.

\subsection{A general framework for the computation of the minimum quantile}
\label{subsec:2:worst_case_dependence}

To be general, let us consider a computer code which takes a vector $\bx \in \chi \subseteq \RR^d$ as an input and produces a real quantity $y$ in output. This code is represented by a deterministic function $\eta: \RR^d \rightarrow \RR$ such that $\eta(\bx) = y$. The sets $\RR$ and $\RR^d$ are endowed with their Borel sigma algebras and we assume that $\eta$ is measurable. The general expression of the function $\eta$ is unknown but for some vector $\bx \in \RR^d$ it is assumed that the quantity $\eta(\bx)$ can always be computed. In particular, the derivatives of $\eta$, when they exist, are never assumed to be known. Let $P_1, \dots, P_d$ be a fixed family of $d$ distributions, all supported on $\RR$. We introduce the set $\cP(P_1, \dots, P_d)$ of all multivariate distributions $P$ on $\RR^d$ such that the marginal distributions of $P$ are all equal to the $(P_j)_{j=1\dots d}$. Henceforth, we use the shorter notation $\cP$ for $\cP( P_1, \dots, P_d)$. 

For some $P \in \cP$, let $G$ be the cumulative distribution function of the model output. In other terms $dG$ is the push-forward measure of $P$ by $\eta$. For $\alpha \in (0,1)$, let $G^{-1}$ be the $\alpha$-quantile of the output distribution: 
\begin{equation}
	G ^{-1} (\alpha) := \inf \{y \in \RR: G (y) \geq \alpha \}.
	\label{eq:2:quantile_function}
\end{equation}
For the rest of this document, we denote as \textit{output quantile} the $\alpha$-quantile of the output distribution. 

In many real situations, the function $\eta$ corresponds to a known physical phenomenon. The input variables $\bx$ of the model are subject to uncertainties and are quantified by the distribution $P$. The propagation of these uncertainties leads to the calculation of the output quantile, which defines an overall risk. Due to the difficulties to gather information, it is common to have this distribution incompletely defined and only known through its marginal distributions. Therefore, the set $\cP$ corresponds to all the possible distributions that are only known through their marginal distributions $(P_j)_{j=1\dots d}$. In a reliability study, it is essential to avoid underestimating the risk. In such a situation, we might consider a more pessimistic computation of the quantile. We define as the \textit{worst quantile}, the minimum value of the quantile by considering all the possible input distributions $P \in \cP$. This conservative approach consists in minimizing $G^{-1}(\alpha)$ over the family $\cP$ such as 
\begin{equation} 
	\Qastar :=  \min_{P \in \cP} \Qa.
	\label{eq:2:general_problem}
\end{equation}
Since the function $\eta$ has no closed form in general, it is not possible to give a simple expression of $G^{-1}(\alpha)$ in function of the distribution $P$, and consequently the minimum $\Qastar$ does not have a simple expression too. In this paper we propose to study a simpler problem than \eqref{eq:2:general_problem}, by minimizing $\Qa$ over a subset of $\cP$. This subset is a family of distributions $(\Ptheta)_{\btheta \in \cD}$ associated to a parametric family of copula $(\Ctheta)_{\btheta \in \cD}$, where $\bTheta$ is a compact set of $\RR^p$ and $p$ is the number of copula parameters.

\subsection{Copula-based approach}

We introduce the real-values random vector $\bX = (X_1, \dots, X_d) \in \RR^d$ associated to the distribution $\Ptheta$. Each component $X_j$, for $j=1, \dots, d$, is a real-value random variable with distribution $P_j$. A copula describes the dependence structure between a group of random variables. Formally, a copula is a multidimensional continuous cumulative distribution function (CDF) linking the margins of $\bX$ to its joint distribution. Sklar's Theorem \cite{Sklar59} states that every joint distribution $\Ftheta$ associated to the measure $\Ptheta$ can be written as
\begin{equation}
	\Ftheta (\bx) = \Ctheta \left(F_1 (x_1), \dots, F_d (x_d) \right),
	\label{eq:2:copula_function}
\end{equation}
with some appropriate $d$-dimensional copula $\Ctheta$ with parameter $\btheta \in \bTheta$ and the marginal CDF's $F_j(x_j) = \PP [X_j \leq x_j]$. If all marginal distributions are continuous functions, then there exists an unique copula satisfying
$$
	\Ctheta (u_1, \dots, u_d) = \Ftheta ( F_1^{-1}(u_1), \dots, F_d^{-1}(u_d))
$$
where $u_j = F_j(x_j)$. For $\Ftheta$ absolutely continuous with strictly increasing marginal distributions, one can derive \eqref{eq:2:copula_function} to obtain the joint density of $\bX$:
\begin{equation}
	\ftheta (\bx) = \ctheta \left(F_1 (x_1), \dots, F_d (x_d) \right) \prod_{j=1}^d f_j(x_j),
	\label{eq:2:copula_density}
\end{equation}
where $\ctheta$ denotes the copula density function of $\Ctheta$ and $f_j(x_j)$ are the marginal densities of $\bX$. Numerous parametric copula families are available and are based on different dependence structures. Most of these families have bidimensional dependencies, but some can be extended to higher dimensions. However, these extensions have a lack of flexibility and cannot describe all types of dependencies \cite{Nelsen07}. To overcome these difficulties, tools like vine copulas \cite{Joe94} (described in Section \ref{sec:quantile_minimization}) combine bivariate copulas, from different families, to create a multidimensional copula.

Let $\Gtheta$ and $\Gtheta^{-1}$ be respectively the CDF and quantile function of the push-forward distribution of $\Ptheta$ by $\eta$ (see  Figure \ref{fig:2:push_forward}). For a given parametric family of copula $(\Ctheta)_{\btheta \in \bTheta}$ and a given $\alpha\in (0, 1)$, the minimum output quantile for a given copula is defined by 
\begin{equation}
	{G ^{-1}_C}^\star  (\alpha) := \inf_{\btheta \in \bTheta} \Gtheta ^{-1} (\alpha)
	\label{eq:2:quantile}
\end{equation}
and if it exists, we consider a minimum
\begin{equation}
	\btheta_{C}^* \in  \argmin_{\btheta \in \bTheta} \Gtheta ^{-1} (\alpha).
	\label{eq:2:minimization_parametric}
\end{equation}
We call this quantity the minimum quantile parameter or worst dependence structure.
\begin{figure}
	\centering
	\includegraphics[ width=0.50\textwidth]{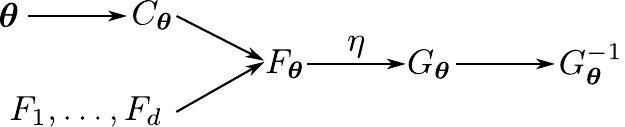}
	\caption{Illustration of the link between the dependence parameter $\btheta$ and the quantile function $\Gtheta^{-1}$. The joint CDF $\Ftheta$ is obtained using \eqref{eq:2:copula_function} from a copula $\Ctheta$ and marginal CDF's $(F_j)_{j=1}^{d}$. The push-forward of $\Ftheta$ through the model $\eta$ leads to the CDF $\Gtheta$ and quantile function $\Qtheta$ of the output distribution.}
	\label{fig:2:push_forward}
\end{figure}
	
Note that there is no reason for $\Gtheta ^{-1} (\alpha)$ to be a convex function of $\btheta$. The use of gradient descent algorithms is thus not straightforward in this context. Moreover, the gradient of $\btheta\to\Gtheta ^{-1}$  is unknown and only  zero-order optimization methods can be applied to solve \eqref{eq:2:minimization_parametric}. For this reason, in the following of this section, we analyze the basic approach which consists in estimating $\btheta_{C}^*$ by approximating $\bTheta$ with a finite regular grid $\bTheta_N$ of cardinality $N$. Therefore, for a given parametric copula $(\Ctheta)_{\btheta \in \bTheta}$ and a given $\alpha \in (0, 1)$, we restrict the problem \eqref{eq:2:minimization_parametric} to
\begin{equation}
	\btheta_{N}^* \in \argmin_{\btheta \in \bTheta_N} \Gtheta ^{-1} (\alpha).
	\label{eq:2:minimization_parametric_grid}
\end{equation}

\subsection{Estimation with a grid search strategy}
\label{subsec:2:grid-search}
 
In the restricted problem \eqref{eq:2:minimization_parametric_grid}, the greater $N$, the closer $\btheta_{N}^*$  to the minimum $\btheta^*$ of $\bTheta$; obviously the convergence rate should depend on the regularity of the function $\eta$ and on the regularity of the quantile function $\btheta \mapsto \Gtheta ^{-1} (\alpha)$. Because $\eta$ has no closed form, the quantile function $\Gtheta ^{-1}(\alpha) $ has no explicit expression. The minimizer  $\btheta_{N}^*$  can be estimated by coupling the simulation of independent and identically distributed (i.i.d) data $(Y_1, \dots, Y_n)$, defined as realizations of the model output random variable $Y := \eta (\bX)$ with distribution $\text{d} \Gtheta$, with a minimization of the empirical quantile over $\bTheta_N$.

For $\btheta$ taking a value over the grid $\bTheta_N$, the empirical CDF of $Y$ is defined for any $y \in \RR$ by
\begin{equation}
	\widehat G_{\btheta} (y) = \frac{1}{n} \sum_{i=1}^n \mathbb 1_{Y_i \leq y}.
	\label{eq:2:empirical_distribution_function}
\end{equation}
The corresponding empirical quantile function $\Qtha$ is defined as in \eqref{eq:2:quantile_function} by replacing $G$ with its empirical estimate. For a given $\alpha$, the worst quantile on the fixed grid $\bTheta_N$ is given by
\begin{equation*}
	\min_{\btheta \in \bTheta_N} \Gtheta ^{-1} (\alpha).
\end{equation*}
and can be estimated by replacing the quantile function with its empirical function:
\begin{equation}
	\min_{\btheta \in \bTheta_N} \Qtha.
	\label{eq:2:extremum_quantile_over_grid}
\end{equation}
Finally the estimation of the minimum quantile parameter over the grid $\bTheta_N$ is denoted by
\begin{equation}
	\hthetaN = \argmin_{\btheta \in \bTheta_N} \Qtha.
	\label{eq:2:extremum_estimation_parametric_grid}
\end{equation}
The construction of the grid $\bTheta_N$ can be difficult because $\bTheta$ can be unbounded (e.g. $\bTheta = [1, \infty]$ for a Gumbel copula). To tackle this issue, we chose to construct $\bTheta_N$ among a normalized space using a concordance measure, which is bounded in $[-1, 1]$ and does not rely on the marginal distributions. We chose the commonly used Kendall rank correlation coefficient (or Kendall's tau) \cite{kendall1938new} as a concordance measure to create this transitory space. This non-linear coefficient $\tau \in [-1, 1]$ is related to the copula function as follows:
$$
	\tau = 4 \int_{-1}^1 \int_{-1}^1 C_{\theta}(u_1, u_2) \mathrm d C(u_1, u_2) - 1.
$$
For many copula families, this relation is much more explicit (see for instance \cite{frees1998understanding}). Therefore, the finite grid is created among $[-1, 1]^p$ and each element of this grid is converted to the copula parameter $\btheta$. Moreover, the use of concordance measures gives a normalized expression of the strength of dependencies for all pairs of variables, independently of the used copula families.

The consistency of estimators \eqref{eq:2:extremum_quantile_over_grid} and \eqref{eq:2:extremum_estimation_parametric_grid} is studied in next section, under general regularity and geometric assumptions on $\eta$ and the functional $\btheta \mapsto \Ptheta$.
	
\subsection{Consistency of  worst quantile-related estimators}
\label{subsec:2:consistency}

In this section, we give consistency results of the estimators $\min_{\btheta \in \bTheta_N} \Qtha$ and $\hthetaN$, for a growing sequence of grids on the domain $\bTheta$. For easier reading, we skip some definitions needed for our assumptions. Section \ref{sec:proof_consitency} in Appendix provides a more complete presentation, including the formal definition of the modulus of increase 
of the quantile function. 

Let $\alpha$ be a fixed value in $(0,1)$. To approximate $\cD $, we consider a sequence of finite discrete grids $(\bTheta_{N})_{N \geq 1}$ on $\cD $ where $N$ is the cardinal of $\bTheta_{N}$ and such that
\begin{equation}
	\sup_{\btheta \in \bTheta ,  \, \btheta'    \in \bTheta_N}    \| \btheta -  \btheta'  \|_2     \rightarrow 0  \textrm{ \quad as $N$ tends to infinity.}
	\label{eq:2:approx_grid}
\end{equation}
We first introduce technical hypotheses required for the consistency result which are commented further in the text. 
\setcounter{assumption}{0}
\begin{assumption}
	\label{hyp:2:continuity_copula}
	For all $\btheta \in \bTheta$, the distribution $\Ptheta$ admits a density $\ftheta$ for the Lebesgue measure and the copula $\Ctheta$ admits a density $\ctheta$ for the Lebesgue measure on $[0, 1]^d$ such that 
	\begin{align*}
		\bTheta \times [0, 1]^d &\longrightarrow \RR \\
		\btheta \times (x_1, \dots, x_d) &\longrightarrow \ctheta (x_1, \dots, x_d)
	\end{align*}
	is a continuous function.
\end{assumption}
\let\theassumption\origtheassumption
\begin{assumption}
	\label{hyp:2:continuous_distribution_D}
	For all $\btheta \in \bTheta$, $\Gtheta$ is a continuous function.
\end{assumption}
\let\theassumption\origtheassumption
\begin{assumption}
	\label{hyp:2:increasing_distribution_D}
	For all $\btheta \in \bTheta$, $\Gtheta$ is strictly increasing and the modulus of increase 
	of $\Gtheta$ at $\Qta$ is lower bounded by a positive function $\underline{ \epsilon}_{\bTheta}$.
\end{assumption}
\begin{assumption}
	\label{hyp:2:min_unicity_D}
	There exists an unique $\btheta^* \in \bTheta$ minimizing $ \btheta \mapsto \Qta$.
\end{assumption}
\setcounter{assumption}{3}

Let $(N_n)_{n\geq 1}$  be a sequence of integers such that  $ N_n \lesssim n^{\beta} $ for some $\beta >0$. For every $n \geq 1$ we consider the grid $\bTheta_{N_n}$ and for every $\btheta \in \bTheta_{N_n}$ we compute the empirical quantile $\Qtha$ from a sample of $n$ i.i.d variables $Y_1,\dots, Y_n $ with $Y_i = \eta(\bX_i)$, where the $\bX_i's$ are i.i.d. random vectors with distribution $\Ptheta$.  We then introduce the extremum estimator 
\begin{equation}
	\label{eq:2:estim}
	\htheta :=  \widehat {\btheta}_{N_n} .
\end{equation}
\begin{theor}
	\label{theor:2:consistency_extremum_estimator}
	Under Assumptions \ref{hyp:2:continuity_copula}, \ref{hyp:2:continuous_distribution_D} and \ref{hyp:2:increasing_distribution_D}, for all $\varepsilon > 0$ we have
	\begin{equation} 
		\label{eq:2:cv_quantile}
		P \left(   \left| \widehat G_{\hat \btheta} ^{-1} (\alpha) -   {G ^{-1}_C}^\star  (\alpha)   \right| > \varepsilon \right)  \xrightarrow{n \rightarrow \infty} 0.
	\end{equation}
	Moreover, if Assumption~\ref{hyp:2:min_unicity_D} is also satisfied, then for all $h > 0$ we have
	$$
		\PP [|\htheta - \btheta_C^*| > h] \xrightarrow{n \rightarrow\infty} 0
	$$
	\emph{(proof given in Appendix ~\ref{sec:proof_consitency})}.
\end{theor}

It would be possible to provide rates of convergence for this extremum quantile and for $\btheta^\star$ at the price of more technical proofs, by considering also the dimension metric of the domain $\bTheta$ and the modulus of increase  of the function $\btheta \mapsto G_{\btheta}(\alpha) $ (see for instance the proofs of Theorems 1 and 2 in \cite{chazal2015rates} for an illustration of such computations). It would be also possible to derive similar results for alternative extremum quantities. One first example, useful in many applications, would be to estimate some risk probability by determining an extremum $\inf_{\btheta \in \bTheta} \Gtheta(y)$ of the CDF for a fixed $y$.

This consistency result could also be extended for regular functional of $\Gtheta$ or $\Gtheta^{-1}$, such that 
$$
	\inf_{\btheta \in \bTheta} \int_{y \geq y_0}  \Gtheta(y) \mathrm{d}y \quad \mathrm{or} \quad 
	\inf_{\btheta \in \bTheta} \int_{\alpha \geq \alpha_0}  \Qta \mathrm{d}y, 
$$
for some fixed values $y_0$ and $\alpha_0$. Extending our results for such quantities is possible essentially because the Dvoretsky-Kiefer-Wolfowitz (DKW) inequality \cite{Dvoretzky56}, used in the proof, gives an uniform control on the estimation of the CDF and the quantile function.

We now discuss the three first assumptions and provide some geometric and probabilistic interpretations of them. Assumption \ref{hyp:2:continuity_copula} requires some regularity of the input distribution with respect to $\btheta$. This is indeed necessary to locate the minimum of the quantile. Assumption \ref{hyp:2:continuous_distribution_D} and \ref{hyp:2:increasing_distribution_D} ensure that the output quantile function $\Qt$ has a regular behavior in a neighborhood of the computed quantile $\Qta$. Assumption \ref{hyp:2:continuous_distribution_D} ensures that the output distribution $\mathrm d \Gtheta$ has no Dirac masses whereas Assumption \ref{hyp:2:increasing_distribution_D} ensures that there is no area of null mass inside the domain of $\mathrm d \Gtheta$.

Figure \ref{fig:2:deviation_increase_CDF} illustrates Assumption \ref{hyp:2:continuous_distribution_D} with a possible configuration of the input distribution. For $\btheta \in \bTheta$ an  $\delta > 0 $, we consider a small neighborhood $[\Qta - \delta, \Qta + \delta]$ of $\Qta $, and the pre-image of this neighborhood. The two right figures are the CDF $\Gtheta$ (top) and PDF $\gtheta$ (bottom) of the output variable $Y$ for a given $\btheta$. The figure at the left hand represents the contours of the pre-image in the input space. The red plain line is the level set $ \eta^{-1} (\Qta)$ and the dot blue line is the perturbed level set $ \eta^{-1}(\Qta \pm \delta))$. The blue area in the right figure corresponds to $[\Qta - \delta, \Qta + \delta]$ and the pre-image of this neighborhood is the blue area in the left figure. Assumption \ref{hyp:2:continuous_distribution_D} requires that the mass of the blue domain is lower bounded by a positive function $\underline{\varepsilon}_{\bTheta}( \delta)$ that does not depend on $\btheta$.

\begin{figure}
	\centering
	\includegraphics[width=1.\textwidth]{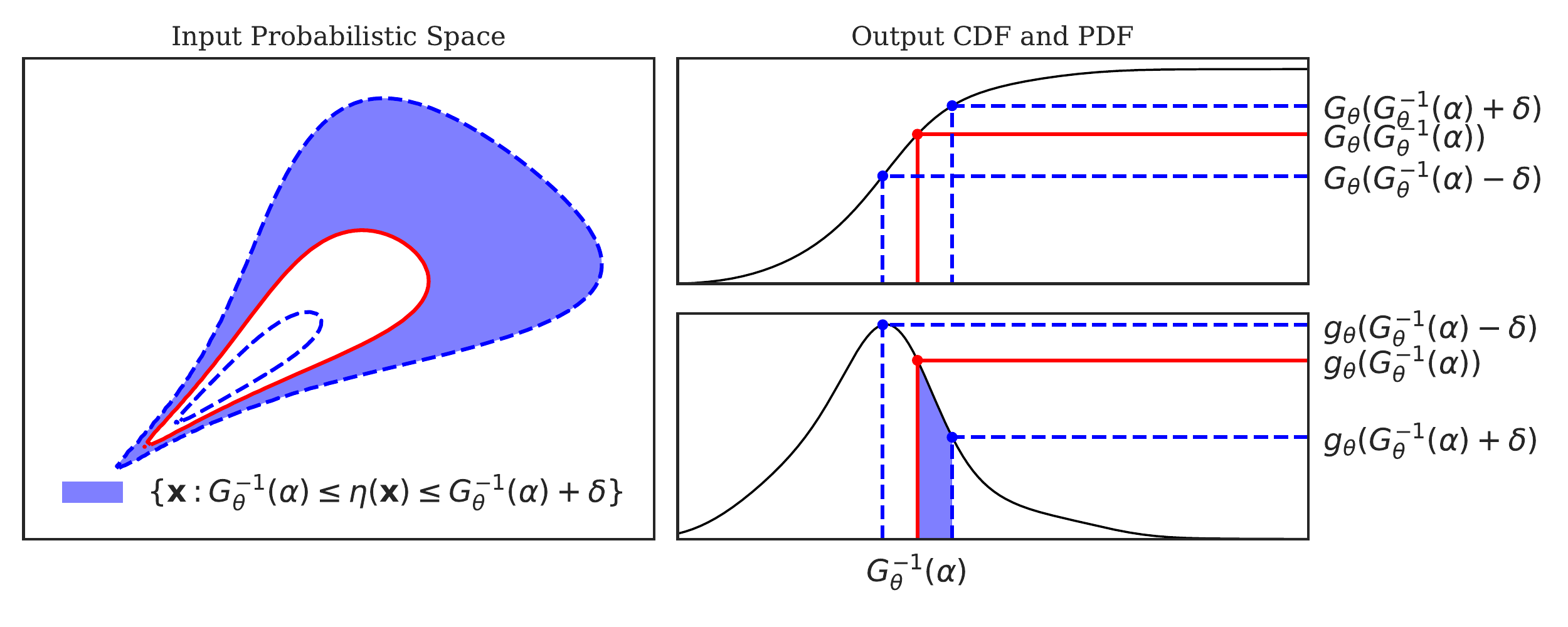}
	\caption{Pre-image (left) and image (right) of a modulus of increase of $\Gtheta$ at the point $\Gtheta ^{-1} (\alpha)$ for a deviation $\pm \delta$.}
	\label{fig:2:deviation_increase_CDF}
\end{figure}

It is possible to give sufficient conditions on the input distribution $\Ftheta$ and on the geometry of the code $\eta$ to obtain Assumptions \ref{hyp:2:continuous_distribution_D} and \ref{hyp:2:increasing_distribution_D}. Using the definition of the modulus of continuity from Equation \eqref{eq:appx:local_modulus_increase} in Appendix, it comes
\begin{eqnarray*}
	\epsilon_{\Gtheta}(\delta, \Qta ) &=&  \max \left[ \int_{\{ \Qta \leq g \leq \Gtheta^{-1}(\alpha) + \delta \}} \ftheta (\bx) \mbox d \lambda (\bx) ;
	\int_{\{ \Gtheta^{-1}(\alpha) - \delta \leq g \leq \Qta  \}} \ftheta (\bx) \mbox d \lambda (\bx) \right] \\ 
	& \geq & 	\int_{\{ \Gtheta^{-1}(\alpha)  \leq g \leq \Qta + \delta  \}} \ftheta (\bx) \mbox d \lambda (\bx) 	
\end{eqnarray*}
Assume that the code $\eta$ is a Lipschitz and differentiable function with no null derivatives almost everywhere in the neighborhood of $\Gtheta^{-1}(\alpha)$. Then, using the \textit{coarea formula} (see for instance \cite{evans2015measure}, Section 3.4.4, Proposition 3), we find that
$$	
\epsilon_{\Gtheta}(\delta, \Gtheta^{-1} (\alpha) ) \geq 
	\int_{\Gtheta^{-1} (\alpha)}^{\Gtheta^{-1}(\alpha) + \delta} \left[ \int_{\eta^{-1} \{ u \}} \frac{\ftheta}{ \|\nabla \eta \|} \mbox d \mathcal H^{d-1}\right] \mbox d u,
$$
where $\mathcal H^{d-1}$ is the $d-1$ dimensional Hausdorff measure (see for instance Chapter 2 in \cite{evans2015measure}). If the copula and the code are such that there exists a constant $I$ such that for any $\btheta \in \cD$ and any $u$ in the support of $\mathrm d \Gtheta$
$$
	\int_{\eta^{-1} \{ u \}} \ftheta   \, \mbox d \mathcal H^{d-1}  \leq I ,
$$
then we find that 
$$
	\epsilon_{\Gtheta}(\delta, \Gtheta^{-1} (\alpha) ) \geq  \delta \frac I {\|\nabla \eta \|_\infty}.
$$
Note that $ \|\nabla \eta \|_\infty  < \infty $ since $\eta$ is assumed to be Lipschitz. We have proved that Assumption~\ref{hyp:2:increasing_distribution_D} is satisfied in this context. Finally, by rewriting again the co-area formula for $\Gtheta(y)$, we find that Assumption~\ref{hyp:2:continuous_distribution_D} is satisfied as soon as the set of stationary points ($\|\nabla \eta (x)\| = 0$) of all level set  $\eta^{-1} \{ u \}$  has null mass for the Hausdorff measure.

In conclusion, we see that for smooth copulas, Assumptions~\ref{hyp:2:increasing_distribution_D} and \ref{hyp:2:continuous_distribution_D} mainly depend on the regularity of the code, by requiring on one side that $\eta$ does not oscillate to much and on the other side that the set of stationary points does not have a positive mass on the level sets of $\eta$.
\section{A preliminary study of the copula influence on quantile minimization}
\label{sec:methodology}

This section is dedicated to a preliminary exploration of the influence of copula structure on the behavior of the worst quantile, illustrated with toy examples. Especially, while it could be expected that $\Qtheta(\alpha)$ is a monotonic function with $\btheta$, and that the minimum can be reached for a trivial copula (i.e., reaching the \FH bounds). Our experiments show that this behavior is not systematic.

\subsection{About the copula choice}

One of the most common approaches to model the dependence between random variables is to assume linear correlations feeding a Gaussian copula. In this case, the problem is reduced by determining the correlation matrix of $\bX$ that minimizes $\Qtheta(\alpha)$. However, the positive semi-definite constraint on the correlation matrix makes the exploration difficult and the minimization harder when the problem dimension increases. Moreover, such a Gaussian assumption is very restrictive and is inappropriate for simulating heavy tail dependencies \cite{malevergne2003testing}. Still in this elliptical configuration, the $t$-copulas \cite{demarta2005t} can be used to counterpart these problems. Nevertheless, tail dependencies are symmetric and with equal strengths for each pair of variables. Another alternative is to consider multivariate Archimedean copulas \cite{mcneil2009multivariate} which are great tools to describe asymmetric tail dependencies. However, only one parameter governs the strength of the dependence among all the pairs, which is very restrictive and not flexible in high dimension. For a same correlation measure between two random variables, multiple copulas can be fitted and lead to a different distribution of $Y$.

It is clear that the copula choice of $\bX$ has a strong impact on the distribution of $Y$ (see for instance \cite{Tang2013impactcopula}). Therefore, various copula types should be tested to determine the most conservative configuration. In the following, we may consider a flexible approach setting by modeling the input multivariate distribution using regular vine copulas (R-vines). The necessary basics of R-vines are introduced in Section \ref{subsec:4:rationale} and detailed in Appendix \ref{def:appx:r_vine}.

\subsection{About the monotony of the quantile}
\label{subsec:3:monotonicity}

For many simple case studies case studies, the worst quantile is reached for perfect dependencies (\FH bounds). More generally, when the function has a monotonic behavior with respect to many variables, it is likely that the minimum output quantile is reached at the boundary of $\bTheta$. This phenomenon is observed for various physical systems.

To illustrate this phenomenon, we consider a simplified academic model that simulates the overflow of a river over a dike that protects industrial facilities. The river overflow $S$ is described by
\begin{equation}
	S = H_d + C_b - Z_v - H  \qquad \text{with} \qquad H = \left( \frac{Q}{B K_s\sqrt{\frac{Z_m - Z_v}{L}}}\right )^{0.6},
	\label{eq:3:flood_model}
\end{equation}
such as, when $S < 0$, a flooding occurs. The involved parameters of \eqref{eq:3:flood_model} are physical characteristics of the river and the dike (e.g., flow rate, height of the dike) which are described by random variables with known marginal distributions. See \cite{iooss2015review} for more information. For a given risk $\alpha$, we aim at quantifying the associated overflow's height describe by the $\alpha$-quantile of $S$. We extend this model by supposing that the friction (Strickler-Manning) coefficient $K_s$ and the maximal annual flow rate $Q$ are dependent with an unknown dependence structure. To show the influence of a possible correlation between $K_s$ and $Q$ on the quantile of $S$, we describe their dependence structure with multiple copula families. 

The Figure \ref{fig:3:flood_example_quantile_variation} shows the variation of the estimated quantile of $S$ (with a large sample size) in function of the Kendall coefficient $\tau$ between $K_s$ and $Q$ for different copula families. We observe different slopes of variation for the different copula families, with lower quantile values for the copulas with heavy tail dependencies (i.e., Clayton, Joe). At independence ($\tau=0$) and for the counter-monotonic configuration ($\tau=-1$),the quantile values of these families are obviously equivalent. This variation is slight and the quantile is still above zero, but this shows how the dependencies can influence the results of a reliability problem. This illustration shows that the minimum is reached at the boundary of the exploration space, where the two variables are perfectly correlated. 
\begin{figure}
	\centering
	\includegraphics[width=0.85\textwidth]{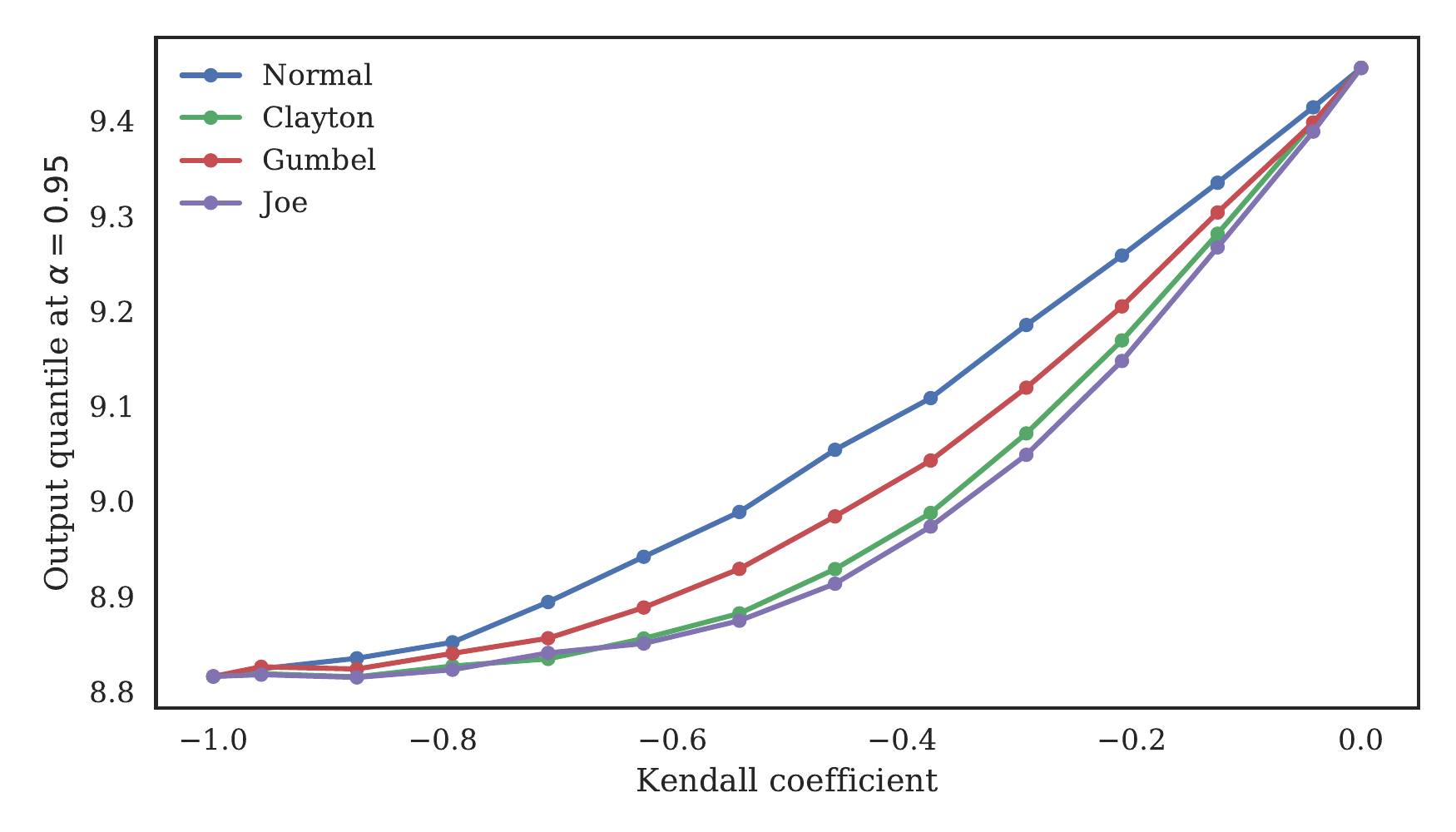}
	\caption{Variation of the quantile of the overflow distribution with the Kendall coefficient $\tau$ for $\alpha=95 \%$ and different copula families (Gaussian, Clayton, Gumbel and Joe).}
	\label{fig:3:flood_example_quantile_variation}
\end{figure}

We can take advantage of this observation to speed up the algorithms presented in the next sections by exploring only the boundaries of $\bTheta$. However, assuming that the minimum is reached on the boundary of $\bTheta$ is a strong assumption that can be unsatisfied in some applications. See Fallacy 3 of \cite{embrechts2002correlation} for a highlight of this pitfall. 

To illustrate this statement, we now give a counter example in the bidimensional setting. We assume uniform marginal distributions for the input such that $X_1 \sim \cU(-3, 1)$ and $X_2 \sim \cU(-1, 3)$, and we consider the model function
\begin{equation}
	\eta(x1, x2) = 0.58 x_1^2 x_2^2 - x_1 x_2 - x_1 - x_2 .
	\label{eq:3:toy_example_bidim}
\end{equation}
The same experience as for Figure \ref{fig:3:flood_example_quantile_variation} is established and the results are shown in Figure \ref{fig:3:example_bidim_1}. The slopes of the quantile estimations with the Kendall coefficient, for each copula families, are quite different than the results of Figure \ref{fig:3:flood_example_quantile_variation}. We observe that the quantile is not monotonic with the Kendall coefficient and its minimum is not reached at the boundary, but for $\tau \approx 0.5$. Moreover, the Gaussian copula is the family that minimizes the most the quantile. It shows that copula with tail dependencies are not always the most penalizing.
\begin{figure}
	\centering
	\includegraphics[width=0.85\textwidth]{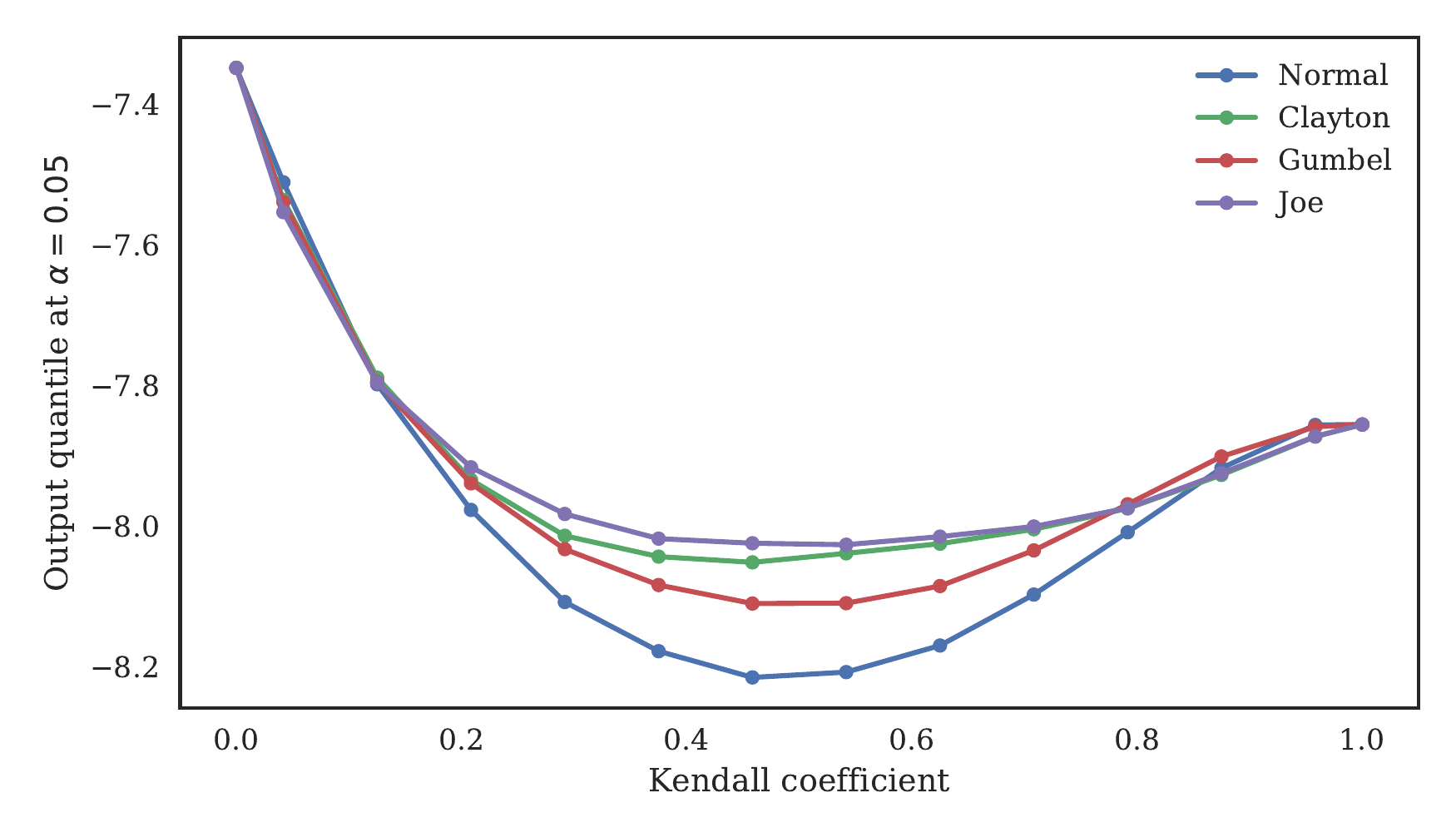}
	\caption{Variation of the output quantile with the Kendall coefficient $\tau$ for $\alpha=5\%$ and different copula families (Gaussian, Clayton, Gumbel and Joe).}
	\label{fig:3:example_bidim_1}
\end{figure}

A second example, inspired from Example 6 of \cite{embrechts2002correlation}, also shows that the worst case dependence structure in an additive problem is not necessary for perfectly correlated variables. We consider a simple portfolio optimization problem with two random variables $X_1$ and $X_2$ with generalized Pareto distributions such as $F_1(x) = F_2(x) = \frac{x}{1+x}$. We aim at maximizing the profit of the portfolio, which is equivalent as minimizing the following additive model function
\begin{equation}
	\eta(X_1, X_2) = -(X_1 + X_2).
\label{eq:3:additive_model}
\end{equation}
We consider the median ($\alpha=0.5$) of the output as an efficiency measure. The Figure \ref{fig:3:example_bidim_2} shows the output median in function of the Kendall coefficient $\tau$ between $X_1$ and $X_2$. Just like the previous example, we observe a non-monotonic slope of the median in function of $\tau$. The variation can be significant and the minimum is obtained at $ \tau \approx 0.53$ for the heavy tail copula families (i.e., Clayton and Joe). The phenomenon can be explained by the marginal distributions of the random variables, which are close Pareto distributions. A large correlation seems to diminish the influence of the tails, which gives a higher quantile value. This explains why the minimum is obtained for a dependence structure other that independence or the perfect dependence.
\begin{figure}
	\centering
	\includegraphics[width=0.85\textwidth]{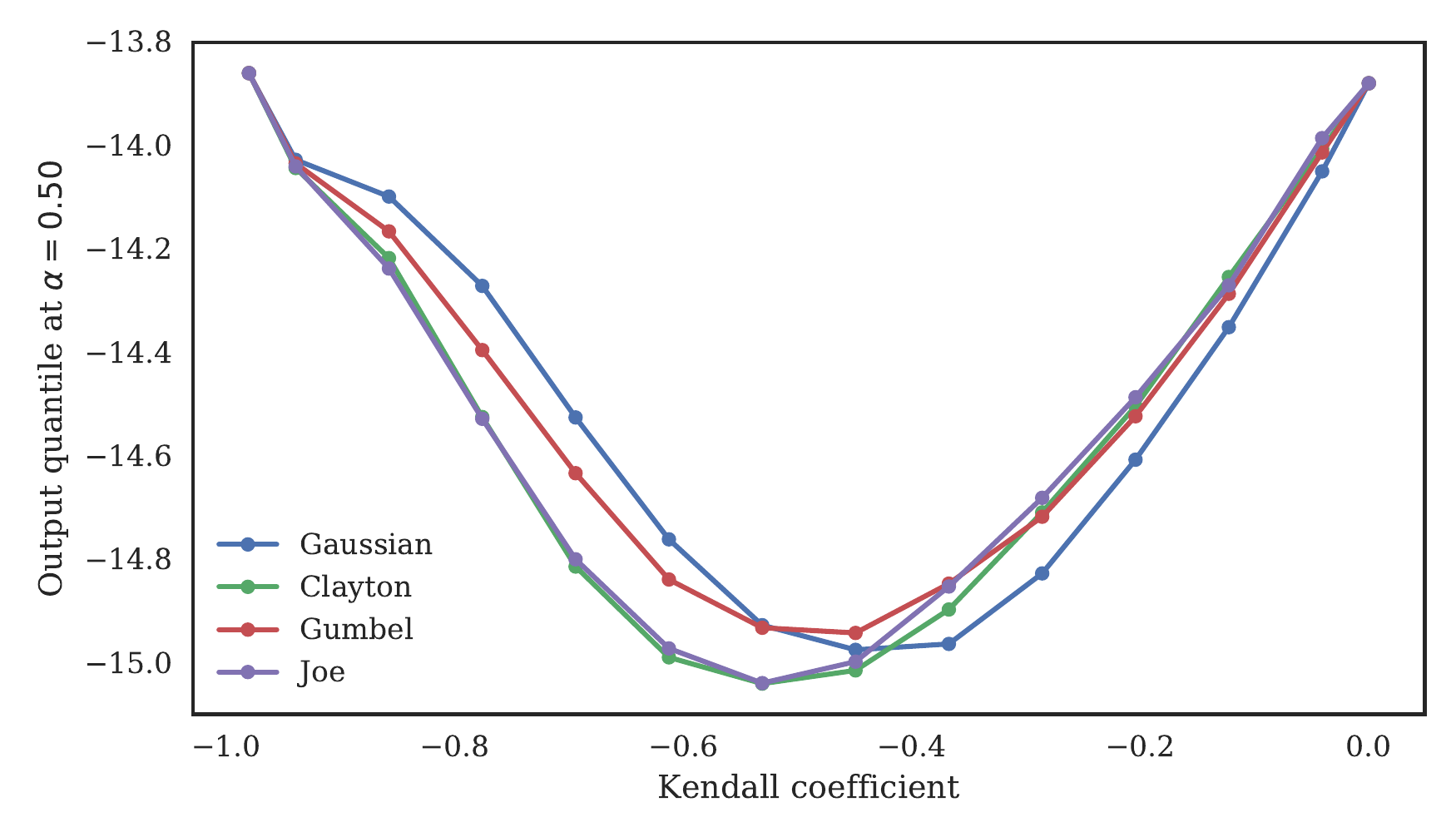}
	\caption{Variation of the portfolio median with the Kendall coefficient $\tau$ for different copula families.}
	\label{fig:3:example_bidim_2}
\end{figure}

Therefore, these examples show that the worst quantile can be reached for other configurations than the perfect dependencies.
\section{Quantile minimization and choice of penalized correlation structure}
\label{sec:quantile_minimization}

This section first provides a rationale for choosing the so-called R-vine structure as a preferential copula structure for modeling the variety of correlations between inputs. Then, the search for a minimum quantile is presented in two times. Subsection \ref{subsec:4:estimation_fixed} proposes an exhaustive grid-search algorithm for estimating this quantile when the R-vine copula structure is fixed with a given pair-copula families and indexed by the parameter vector $\btheta$. Subsection \ref{subsec:4:estimation_iterative} extends this rigid framework by permitting the search of particular sub-copula pairwise structures, such that the minimization be more significant. In each situation, examples are provided to demonstrate the feasibility of the approach.

\subsection{A rationale for R-vine copula structures}
\label{subsec:4:rationale}

Representing multi-dimensional dependence structures in high dimensional settings is a challenging problem. For the following definition, we simplify the expressions by omitting the use of $\btheta$: $f=\ftheta$, $F=\Ftheta$ and $c = \ctheta$. By recursive conditioning, the joint density can be written as a product of conditioning distributions such as
\begin{equation}
	f(x_1, \dots, x_d) = f_1(x_1) \cdot f_{2|1}(x_{2} | x_1) \cdot f_{3 |1, 2}(x_3 | x_1, x_2) \cdots f_{d | 1, 2, \dots d-1}(x_d | x_1, x_2, \dots, x_{d-1}).		
	\label{eq:4:PCC_eq1}
\end{equation}
For clarity reason, we now simplify the expression with $f_{3 |1, 2} = f_{3 |1, 2}(x_3 | x_1, x_2)$ and so on for other orders. From \eqref{eq:2:copula_density}, the conditioning densities of \eqref{eq:4:PCC_eq1} can be rewritten as products of conditioning copula and marginal densities. For example, in a case of three variables and using \eqref{eq:4:PCC_eq1}, one possible decomposition of the the joint density can be written as 
\begin{equation}
	f(x_1, x_2, x_3) = f_1 \cdot f_{2|1} \cdot f_{3 |1, 2}.
	\label{eq:4:PCC_exemple_1}
\end{equation}
Using  \eqref{eq:2:copula_density}, the reformulation of $f_{3 |1, 2}$ leads to 
\begin{equation}
	f_{3 |1, 2} = \frac{f_{1, 3|2}}{ f_{1 | 2}} = \frac{c_{1, 3 |2} \cdot f_{1|2} \cdot f_{3|2}}{f_{1|2}} = c_{1, 3 |2} \cdot f_{3|2}
	\label{eq:4:PCC_exemple_f3}
\end{equation}
where $c_{1, 3 |2} = c_{1, 3 |2}(F_{1 |2} (x_1 | x_2), F_{3|2} (x_3 |x_2))$. By developing $f_{3 |1, 2}$ in the same way, we find that
\begin{equation}
	f_{3 |1, 2} = c_{1, 3 |2} \cdot c_{2, 3} \cdot f_{3}.
	\label{eq:4:PCC_exemple_3}
\end{equation}
Thus, by replacing the expression of $f_{3 |1, 2}$ in \eqref{eq:4:PCC_exemple_1} and doing the same procedure for $f_{2|1}$, the joint density can be written as
\begin{equation}
	f(x_1, x_2, x_3) = f_{1} \cdot f_{2} \cdot f_{3} \cdot c_{1, 2} \cdot c_{2, 3} \cdot c_{1, 3 |2}.
	\label{eq:4:PCC_exemple_4}
\end{equation}
This final representation of the joint density based on pair-copulas has been developed in \cite{joe1996families} and is called the pair-copula construction (PCC). The resulting copula represented by the product of conditional copulas in \eqref{eq:4:PCC_exemple_4} offers a very flexible way to construct high-dimensional copulas. However, it is not unique; indeed, \eqref{eq:4:PCC_eq1} has numerous decomposition forms and it increases with the dimension.

To describe all such possible constructions in an efficient way, \cite{bedford2001probability, bedford2002vines} introduced the vine models. This graphical tool, based on a sequence of trees, gives a specific way to decompose the multivariate probability distribution. Basically, a vine model is defined by
\begin{itemize}
	\item a structure of trees which can be represented by a matrix \cite{morales2010bayesian},
	\item a copula family for each pair of the structure,
	\item a parameter for each pair-copula.
\end{itemize}
A R-vine is the general construction of a vine model, but particular cases exists such as the D-vines and C-vines, described in Appendix \ref{sec:vine_copulas}. Vine models were deeply studied in terms of density estimation and model selection using maximum likelihood \cite{Aas2009mle}, sequential estimation \cite{kurowicka2011estimation, dissmann2013selecting}, truncation  \cite{aas2012truncated} and Bayesian techniques  \cite{Gruber2015bayesian}. Their popularity and well-known flexibility led us to use R-vines in this article, despite the fact that in our context we are looking for a conservative form and not to select the most appropriate form with given data, in absence of correlated observations.

\subsection{Estimating a minimum quantile from a given R-vine}
\label{subsec:4:estimation_fixed}

\subsubsection{Grid-search algorithm}
\label{subsubsec:4:grid_search_algorithm}

Let $\Omega = \{ (i, j) : 1 \leq i, j\leq d \}$ be the set of all the possible pairs of $\bX$, in a $d$-dimensional problem. The number of pairs $p$ is associated to the size of $\Omega$ such as $p = |\Omega| = {{d}\choose{2}} = d (d-1) / 2$. We define $\cV$ as the vine structure and we consider fixed copula families for each pair. In this article, we only consider single parameter pair-copulas, such that the parameter $\btheta$ is a $p$-dimensional parameter vector with a definition space $\bTheta := \prod_{(i, j) \in \Omega } \Theta_{i, j} $ where $\Theta_{i, j}$ is the parameter space for the pair-copula of the pair $(i, j)$. However, the methodology can easily be extended to multi-parameter pair-copulas. Note that a pair-copula can be conditioned to other variables, depending on its position in the vine structure $\cV$. Thus, the input distribution $\mathrm d \Ftheta(\cV)$ is defined by the vine structure $\cV$, the copula families and the parameter $\btheta$. Also note that the copula parameter $\btheta$ is associated to the R-vine structure $\cV$ (i.e., $\btheta = \btheta_{\cV}$), see Section \ref{subsubsec:4:influence_vine}. For the sake of clarity, we simplify the notation to $\btheta$ only.

The most direct approach to estimate the minimum quantile is the Exhaustive Grid-Search algorithm, described by the following pseudo-code.

\begin{algorithm}[H]
\KwData{A vine structure $\cV$, a fixed grid $\bTheta_N$, a sample size $n$}
\For {$\btheta \in \bTheta_N$}{ 
	\textbf{1.} Simulate a sample $\{\bX_i \}_{i=1}^n$ according to $\text d F_{\btheta}(\cV)$\;
	\textbf{2.} Evaluate $\{Y_i = \eta(\bX_i)\}_{i=1}^n$\;
	\textbf{3.} Compute $\Qtha$: empirical quantile of $\{Y_i\}_{i=1}^n$\;
}
\KwResult{$\min\limits_{\btheta \in \bTheta_N} \Qtha$}
\caption{Exhaustive Grid-Search algorithm to minimize the output quantile.}
\label{algo:4:grid_search}
\end{algorithm}

For a given vine structure $\cV$, copula families, a grid $\bTheta_N$ and a sample size $n$, three steps are needed for each $\btheta \in \bTheta_N$. The first step simulates an input sample $\{\bX_i \}_{i=1}^n$ according to the distribution $\mathrm d \Ftheta(\cV)$ for a given sample size $n$. The second evaluates the sample through the model $\eta$. The third estimates the output quantile from the resulting sample $\{Y_i = \eta(\bX_i)\}_{i=1}^n$. The minimum quantile is took among the results of each loop.

\subsubsection{Influence of the vine structure}
\label{subsubsec:4:influence_vine}

Using R-vines, the dependence parameter $\btheta$ is associated to the vine structure $\cV$. Due to the hierarchy of the vine structure, some pair-copulas are conditioned to other variables and thus for their parameters. As an illustration, let us consider two vine structures with the two following copula densities, with the same simplified expressions as for \eqref{eq:4:PCC_exemple_4}:
\begin{align}
	c_{\cV_1}(x_1, x_2, x_3, x_4) &= c_{\theta_{1, 3}}\cdot c_{\theta_{1, 2}} \cdot c_{\theta_{2, 4}} \cdot c_{\theta_{2, 3 | 1}} \cdot c_{\theta_{1, 4 | 2}} \cdot c_{\theta_{3, 4 | 1,2}} \\
	c_{\cV_2}(x_1, x_2, x_3, x_4) &= c_{\theta_{1, 3}}\cdot c_{\theta_{3, 4}} \cdot c_{\theta_{2, 4}} \cdot c_{\theta_{1, 4 | 3}} \cdot c_{\theta_{2, 3 | 4}} \cdot c_{\theta_{1, 2 | 3,4}} .
\end{align}
The difference between these densities is the conditioning of some pairs, the dependence parameters of theses vines are $\btheta_{\cV_1} = [\theta_{1, 2}, \theta_{1, 3}, \theta_{1, 4 | 2}, \theta_{2, 3 | 1}, \theta_{2, 4}, \theta_{3, 4 | 1,2}]$ and $\btheta_{\cV_2} = [\theta_{1, 2 |3, 4}, \theta_{1, 3}, \theta_{1, 4 | 3}, \theta_{2, 3 | 4}, \theta_{2, 4}, \theta_{3, 4}]$. Applying the same grid for these two vines may give different results due to the conditioning order from the vine structure. For example, if the pair $X_3$-$X_4$ is very influential on minimizing the output quantile, it would be more difficult to find a minimum with $\cV_1$ than $\cV_2$ due to the conditioning of the pair with $X_1$ and $X_2$ in $\cV_1$. However, if the grid is thin enough, the minimum from these two vines should be equivalent. 

To counter this difficulty, one possible option consists in randomly permuting the indexes of the variables and repeating the algorithm several times to visit different vines structures.

\subsubsection{Computational cost}

For one given R-vine structure and one fixed copula family at each pair, the overall cost of the method is equal to $nN$. However, as explained in $\S$ \ref{subsec:2:consistency}, the finite grid $\bTheta_N$, should be thin enough to reasonably explore $\bTheta$. Therefore, $N$ should increase with the number of dimensions $d$ and more specifically with the number of pairs $p = {{d}\choose{2}}$. A natural form for $N$ would be to write it as $N = \gamma ^ p$, where $\gamma \in \RR^{+}$. Thus, the overall cost of the exhaustive grid-search would be equal to $n \gamma^{{d}\choose{2}}$. The cost is in $\cO(\gamma^{d^2})$ which makes the method hardly scalable when the dimension $d$ increases.

\subsection{Iterative search for a penalizing R-vine structure: a greedy heuristic based on pairwise copula}
\label{subsec:4:estimation_iterative}

\subsubsection{Going further in quantile minimization}

With Algorithm \ref{algo:4:grid_search}, the previous subsection proposes an exhaustive grid-search strategy to determine a R-vine copula $C_{\tilde\theta}$ such that the associated output quantile $G_{\tilde\btheta}^{-1} (\alpha)$ be the smallest (and also the most conservative in a structural reliability context). This approach remains however limited in practice since $C_{\tilde\theta}$ for fixed pair-copula families (e.g., Archimedean or max-stable copulas) and $\cV$ which is a member of the set ${\cal{F}}_d$ of all the possible $d-$dimensional R-vine structure. Intuitively, a more reliable approach to quantile minimization should be based on mixing this estimation method with a selection among all members of the finite set ${\cal{F}}_d$, as well for the copula families. It is indeed likely that searching within an associative class of copulas like Archimedean ones, allowing modeling dependence in arbitrarily high dimensions, be a too rigid choice for estimating the minimum $G_{\tilde\btheta}^{-1} (\alpha)$.

A minimum quantile can probably be found using a R-vine structure defined by conditional pairwise sub-copulas (according to \eqref{eq:4:PCC_exemple_4}) that are not part of the same rigid structure. However, a brute force exploration of ${\cal{F}}_d$ would be conducted at an exponential cost increasing with $d$ \cite{morales2011counting}. If we also consider the large computational cost of an exhaustive grid-search for a large number of dependent variables (as explained in $\S$ \ref{subsec:4:estimation_fixed}), this approach is not feasible in practice for high dimensions.

For this reason, it is proposed to extend Algorithm \ref{algo:4:grid_search} by a greedy {\it heuristic} that dynamically selects the most influential correlations between variables while limiting the search to pairwise correlations. Doing so, minimizing the output quantile can be conducted in a reasonable computational time. Therefore the selected $d-$dimensional vine structure would be filled with independent pair-copulas except for the pairs that are influential on the minimization.

This working limitation, interpreted as a sparsity constraint, is based on the following assumption: it is hypothesized that only few pairs of variables have real influences on the minimization. It is close in spirit to the main assumption of global sensitivity analysis applied to computer models, according to which only a limited number of random variables has a major impact on the output \cite{saltelli2000sensitivity, iooss2015review}.
 
\subsubsection{General principle}

The method basically relies on an iterative algorithm exploring pairwise correlations between the uniform random variables $U_j=F^{-1}_j(X_j)$ and progressively building a non-trivial R-vine structure, adding one pair of variable to the structure at each iteration. Starting at step $k=0$ from the simple independent copula 
$$
	C_{\btheta^{(0)}}(u_1,\ldots,u_d) = \prod_{j=1}^d u_j,
$$
the algorithm finally stops at a given step $k=K$ while proposing a new copula $C_{\btheta^{(K)}}$ associated to a R-vine structure $\cV_K$ mostly composed of independent pair-copulas.

At each iteration $k$, we denote by $\Omega_k$ the \textit{selected pairs} which are considered non-trivial (non-independent) due to their influence on the quantile minimization. Let $\Omega_{-k} = \Omega \backslash \Omega_k$ be the \textit{candidate pairs}, which were not the remaining pairs, which influence on the minimization is still to be tested and are still considered independent. We also consider $\mathbf B$ as a set of candidate copula families. The pseudo-code of Algorithm \ref{algo:4:iterative_algorithm} shows in detail how this iterative exploration and building is conducted. More algorithms in Appendix \ref{subsec:appx:vine_construction} described how to construct a vine structure with a given list of indexed pairs of variable.

\begin{algorithm}
\vspace{.4em}
\textbf{Initialization: \\}

Iteration: $k = 0$\;
Selected pairs: $\Omega_0 = \emptyset$\;
Selected families: $\mathbf B_0 = \emptyset$\;

\vspace{.4em}
\While{$k \leq K$}{
	Copula parameter space of the selected pairs: $\bTheta_k = \prod_{(i, j) \in \Omega_k} \Theta_{i, j}$\;
	\textbf{1.} Explore the set of candidate pairs $\Omega_{-k}$\;
	\For{$(i, j) \in \Omega_{-k}$}{
		\textbf{a.} Create a vine structure $\cV_{(i, j)}$ using the procedure of Section \ref{subsec:appx:vine_construction} applied to the list $\Omega_k \cup (i, j)$\;
		\textbf{b.} Explore the set of candidate families $\bB$\;
		\For{ $ \cB \in \bB$} {
			Apply Algorithm \ref{algo:4:grid_search} with the pair-copula families $\cB \cup \bB_k$\;
			\begin{enumerate}[(i)]
			\item Define a $(k+1)-$dimensional grid $\bm \Delta_{i,j}$  of $\bTheta_k \times \Theta_{i, j}$ with cardinality $N_k$;
			\item Select the minimum over the grid $\bm \Delta_{i,j}$:
			\begin{eqnarray*}
				\htheta_{\cB} & = & \argmin_{\btheta_{\cB} \in \bm \Delta_{i,j}}\left\{ \widehat{G}_{\btheta_{\cB}}^{-1} (\alpha)\right\}.
			\end{eqnarray*}
			\end{enumerate}
		} % end for loop
		\textbf{c.} Select the minimum among $\mathbf B$
		\begin{eqnarray*}
			\cB_{i, j} & = & \argmin_{\cB \in \mathbf B} \left\{ \widehat{G}_{\htheta_{\cB}}^{-1} (\alpha)\right\} \\
			\htheta_{i,j} & = & \htheta_{\cB_{i, j}}
		\end{eqnarray*}
	} % end for loop
	\textbf{2.} Select the minimum among $\Omega_{-k}$
	\begin{eqnarray*}
		(i,j)^{(k)} & = & \argmin_{(i, j) \in \Omega_{-k}}\left\{ \widehat{G}_{\htheta_{i,j}}^{-1} (\alpha)\right\}, \\
		\cV^{(k)} & = & \cV_{(i, j)^{(k)}}, \\
		\htheta^{(k)} & = & \htheta_{(i,j)^{(k)}} \\
		\cB^{(k)} & = & \cB_{(i,j)^{(k)}}
	\end{eqnarray*}

	\textbf{3.} Check the stopping condition\;

	\If{$\widehat{G}_{\htheta^{(k)}}^{-1} (\alpha) \geq \widehat{G}_{\htheta^{(k-1)}}^{-1} (\alpha)$}{
		$K = k-1$\;
	}
	\Else{
		Extend the list of selected pairs: $\Omega_k = \Omega_k \cup (i, j)^{(k)}$ and families: $\bB_k = \bB_k \cup \cB^{(k)}$ \;

		\If{ $k < K$ \textbf{and} computational budget not reached}{
			New iteration: $k = k + 1$\;
		}
	}
} % end while loop
\caption{Minimization of the output quantile and estimation of $\btheta^{(K)}$ over an increasing family of R-vine structures.}
\label{algo:4:iterative_algorithm}
\end{algorithm}

\subsubsection{Example}
\label{subsubsec:4:example_iterative}

Consider the four-dimensional $(d=4)$ situation such as $\bX = (X_1, \ldots, X_4)$ where, for to the sake of simplicity, all marginal distributions of $\bX$ are assumed to be uniform on $[0,1]$. We consider a simple additive model described by
\begin{equation}
	\eta(\bX) = 30 X_1 + 10 X_3 + 100 X_4.
	\label{eq:4:example}
\end{equation}
For an additive model and uniform margins, the output quantile is monotonic with the dependence parameters (see Section \ref{subsec:3:monotonicity}) which locates the minimum quantile at the edge of $\bTheta$. Thus, Step 1.b. of Algorithm \ref{algo:4:iterative_algorithm} is simplified by considering only \FH copulas in the exploration.

In this illustration we consider $\alpha=0.1$ and we select $n=300,000$ large enough in order to have a great quantile estimation and the algorithm stops at $K=3$. Figure \ref{fig:4:example_iterative} shows, for each iteration $k$, the $p-k$ vine structures that have been created by the algorithm. The red nodes and edges are the candidate pairs $(i, j) \in \Omega_{-k}$ and the blue nodes and edges are the selected pairs $\Omega_k$. At iteration $k=0$, the selected pair is $(1, 4)$ with an estimated minimum quantile of $-52.18$. At iteration $k=1$, the second selected pair is $(3, 4)$ with an estimated minimum quantile of $-56.03$. At iteration $k=2$, the third selected pair is $(2, 4)$ with an estimated minimum quantile of $-56.23$.
\begin{figure}
	\centering
    \makebox[\textwidth][c]{\includegraphics[width=1.35\textwidth]{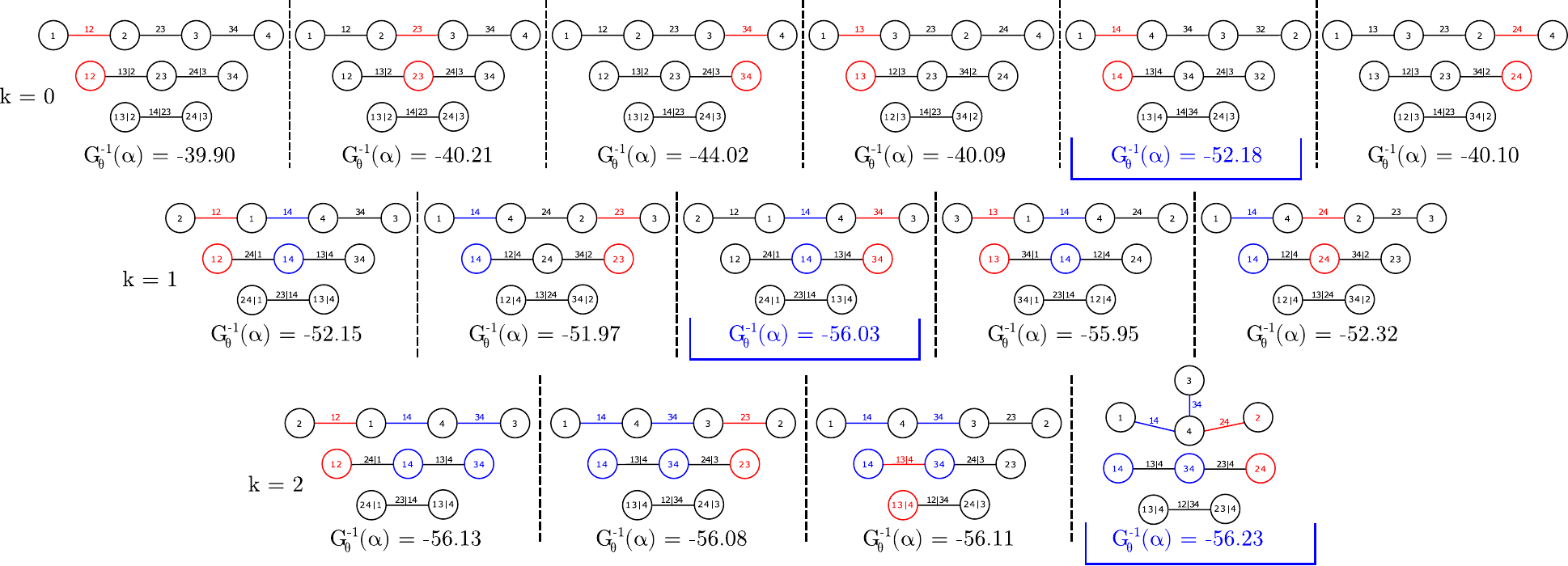}}
	\caption{Illustration of the vine structures created during the 3 iterations of the algorithm for the example of Section \ref{subsubsec:4:example_iterative}. The candidate and selected pairs are respectively represented in red and blue. The quantile associated to the selected pair of each iteration is written in blue.}
	\label{fig:4:example_iterative}
\end{figure}

We observe that $X_4$ appears in all the selected pairs. This is not surprising since $X_4$ is the most influential variable with the largest coefficient in \eqref{eq:4:example}. The algorithm considers D-vines by default, but this is important for the first iterations since most of the pairs are independent. When it is possible, the algorithm creates a vine such as the selected pairs and the candidate pair are in the first trees. For example, the fourth vine at iteration $k=2$ with the candidate pair $(2, 4)$ shows a R-vine structure that respects the ranking of the listed pairs. However, the third vine at iteration $k=2$ for the candidate pair $(1, 3)$ along with the selected pairs $\{ (1, 4), (3, 4)$ could respect the ranking and set all the pairs in the first tree altogether. Thus, using Algorithm \ref{algo:appx:vine_construction} in Appendix, a valid vine structure is determined by placing the candidate pair $(1, 3)$ in the next tree.

\subsubsection{Computational cost}
\label{subsubsec:4:comput_cost}

The number of model evaluations is influenced by several characteristics from the probabilistic model and from the algorithm. Let $|\bB|$ be the number of family candidates. The total number of runs is 
\begin{equation}
	N = |\bB| \frac{n}{2} \sum_{k=0}^{K} N_k \times ( d (d-1)- 2k).
\end{equation}
The sum corresponds to the necessary iterations to determine the influential pairs. The maximum possible cost is if all the pairs are equivalently influential (i.e., $K = p = d (d-1) / 2$), which would be extremely high. The term $n N_k$ is the cost from the grid-search quantile minimization at step \textbf{2.} of the algorithm. The greater $N_k$ is and the better the exploration of $\bTheta_k \cup \Theta_{i, j}$. Because the dimension of $\bTheta_k$ increases at each iteration $k$, it is normal that $N_k$ should also increases with $k$ (e.g. $N_k = \gamma^{ \beta k}$, where $\gamma$ and $\beta$ are constants). Also, the greater $n$ is and the better the quantile estimations. The second term is the cost from the input dimension $d$ which influences the number of candidate pairs $\Omega_{-k}$ at each iteration $k$.

Extensions can be implemented to reduce the computational cost such as removing from $\Omega$, the pairs that are not sufficiently improving the minimization.
\section{Applications}
\label{sec:applications}

The previously proposed methodology is applied to a toy example and a real industrial case-study. It is worth to mention that these experiments (and future ones) can be conducted again using the Python library \texttt{dep-impact} \cite{depimpact}, in which are encoded all the procedures of estimation and optimization presented here.

\subsection{Numerical example}

We pursue and extend the portfolio example considered in Section \ref{subsec:3:monotonicity} and illustrated on Figure \ref{fig:3:example_bidim_2}. The numerical model $\eta$ is now defined by the weighted sum
\begin{equation}
	Y = \eta(\bX) =  - \bbeta \bX^T = - \sum_{j=1}^d \beta_j X_j,
\end{equation}
where the $\bbeta = (\beta_1, \dots, \beta_d)$ is a vector of constant weights. The margins of the random vector $\bX$ follow the same generalized Pareto distribution with scale $\sigma$ and shape parameter $\xi$. Note that the bivariate example in Section \ref{subsec:3:monotonicity} considered $\beta = 1$ and the distribution parameters as $\sigma=1$ and $\xi=1$. In the following examples, we aim at minimizing the median ($\alpha = 0.5$) of the output distribution. We chose to fix the marginal distribution's parameters at $\sigma=10$ and $\xi=0.75$, and we set the constant vector $\bbeta$ to a base-10 increasing sequence such that $\bbeta = (10^{1/d}, 10^{2/d}, \dots, 10)$. This choice of weights aims to give more influence to the latest components of $\bX$ on $Y$. Thus, some pairs of variables should be more important in the minimization of the output quantile, as required by the sparsity constraint. We also took $n$ large enough to estimate the output quantile with high precision (i.e. $n=300,000$).

For all these experiments the results from the different methods can be compared.
\begin{itemize}
\item Method 1: the grid-search approach with an optimized LHS sampling \cite{mckay1979comparison} inside $\bTheta$ and a random vine structure,
\item Method 2: the iterative algorithm with an increasing grid-size of $N_k = 25^*(k+1)^2$.
\end{itemize}

The Method 1 is established with the same computational budget as Method 2.

\subsubsection{Dimension 3}

In a three dimensional problem, only three pairs of variables ($p=3$) are involved in the dependence structure. The sampling size of $\bTheta$ in Method 1 is set to $400$, which is great enough to explore a three dimensional space. The results are displayed on Figure \ref{fig:5:toy_example_dim_3_matrix_plot}: the estimated quantiles from Method 1 (blue dots) with a convex hull (blue dot line) and the quantile at independence (dark point) are provided. It also highlights the minimum estimated quantiles from Methods 1 and 2 which are respectively represented in blue and red points. We also show in green point, the minimum quantile by considering only the \FH bounds. For each minimum, the 95 \% bootstrap confidence intervals is displayed in dot lines.
\begin{figure}
	\centering
	\includegraphics[width=0.95\textwidth]{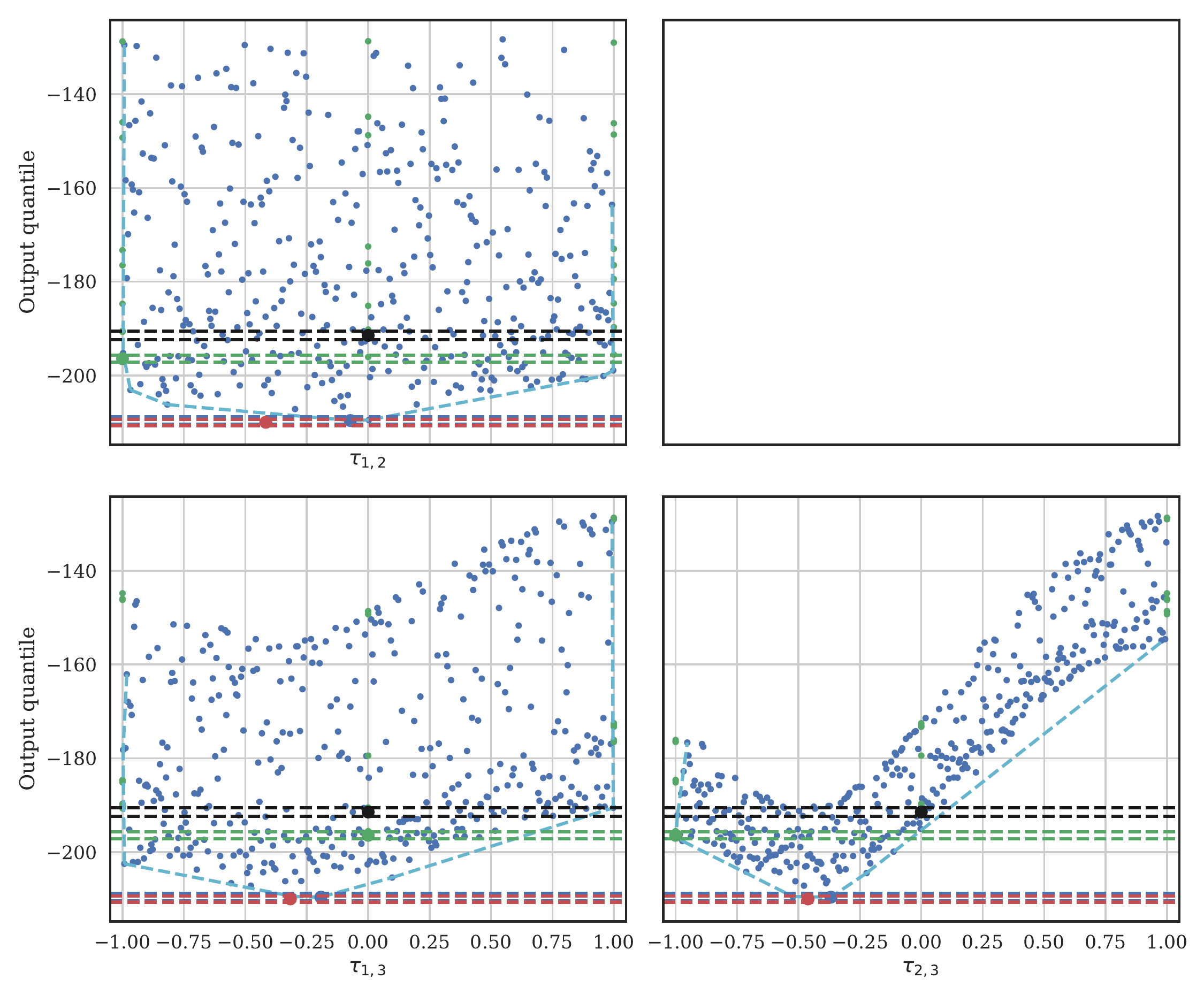}
	\caption{Matrix plot of the output median in function of the Kendall coefficient of each pair. The blue dots represents the estimated quantiles of Method 1. The black point is the quantile at independence and the minimum of Method 1 and 2 are the red and blue points, which are equivalent here. The green point is the the minimum with only \FH bounds. The 95 \% bootstrap confidence intervals are displayed in dot lines.}
	\label{fig:5:toy_example_dim_3_matrix_plot}
\end{figure}

This low dimensional problem confirms the non-monotonic form of the quantile with the dependence parameter, in particular for the variation of the quantile in function of $\tau_{2, 3}$. As expected, the pair $X_2$-$X_3$ is more influential on the output quantile due to the large weights on $X_2$ and $X_3$. The minimum values obtained by each method are still lower that the results given by an independent configuration. The minimum using \FH bounds is also provided to show that the minimum is not at the boundary of $\bTheta$. Method 1 and 2 have very similar minimum results.

\subsubsection{Dimension 10}

To illustrate the advantages of the iterative procedure, we now consider $d=10$. In this example, we chose to only consider a Gaussian family for the set of pair-copula family candidates. The sampling size for the exploration of $\bTheta$ in Method 1 is set to $6,000$. Experimental results are summarized over Figure \ref{fig:5:iterative_dim_10}, by displaying the minimum quantiles in function of the iteration $k$ of Method 2. The quantile at independence is shown in dark line, the minimum estimated quantile from Method 1 is shown in blue line and the other lines are the minimum quantiles at each iteration of the algorithm, all with their 95\% bootstrap confidence interval. We display at each iteration the minimum quantiles of each candidate pair in small dots.
\begin{figure}
	\centering
	\includegraphics[width=0.9\textwidth]{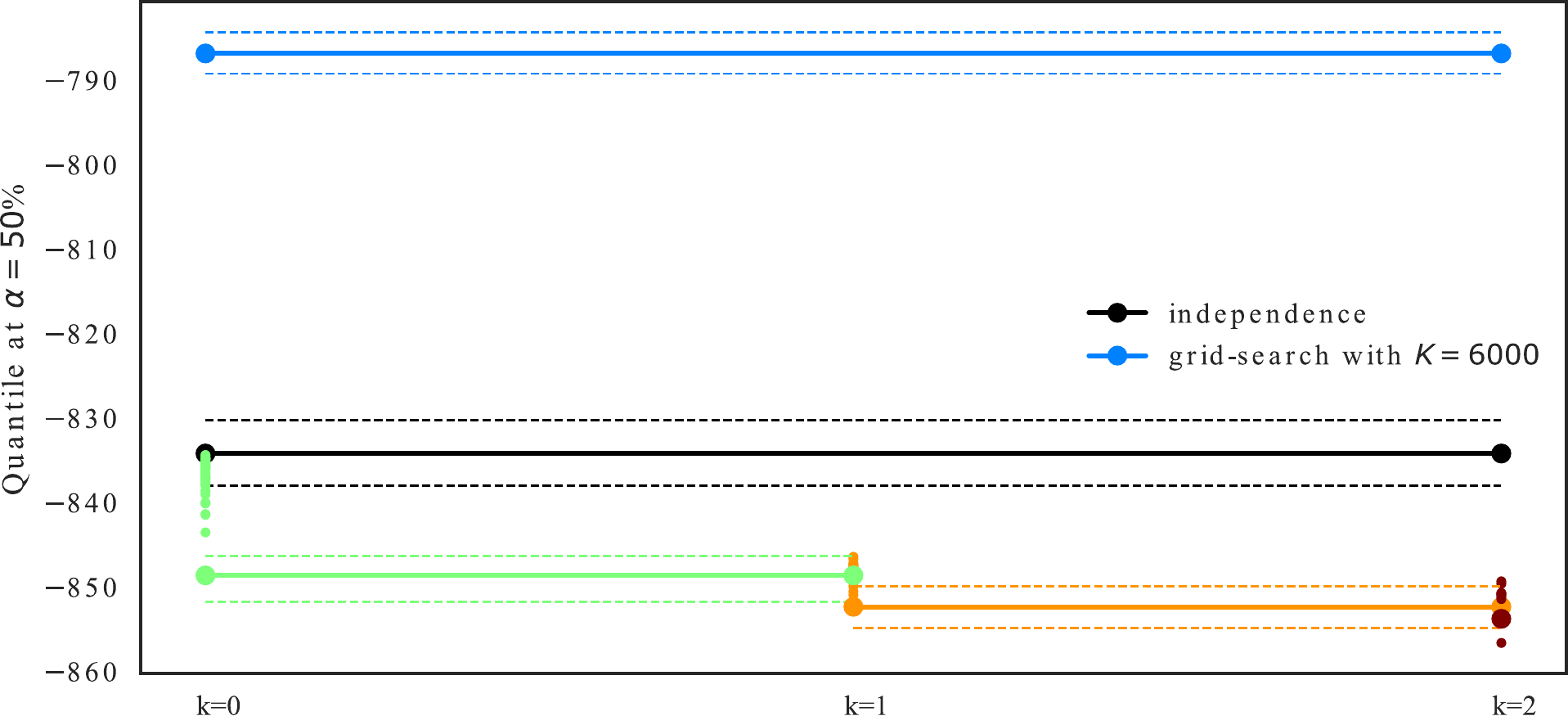}
	\caption{Minimum quantile results with the iteration $k$ of the iterative procedure. The quantile at independence is shown in dark line. The minimum quantiles from Method 1 is show in blue lines. The other lines and dots colors are the results from Method 2. For each iteration, the small dots are the estimated quantiles of all candidates and the point is the minimum. The 95 \% bootstrap confidence intervals are also displayed for the independence and each minimums.}
	\label{fig:5:iterative_dim_10}
\end{figure}

The minimum result from Method 1 is even higher than the quantile at independence. This is due to the very large number of pairs ($p=45$) that makes the exploration of $\bTheta$ extremely difficult. On the other hand, Method 2 (iterative algorithm) is definitely better and significantly decreases the quantile value even at the first iteration (for only one dependent pair). The results are slightly improved with the iterations. We observe at the last iteration that the results from the candidate pairs are slightly higher than the minimum from the previous iteration. It is due to the choice of $N_k$ which does not increases enough with the iterations to correctly explore $\bTheta_k$, which also increases with the iterations.

\subsubsection{Using multiple pair-copula family candidates}

To show the importance of testing multiple copula families, we consider $d=6$ and three tests of Method 2 (iterative procedure). The Figure \ref{fig:5:iterative_multi_k_3} shows the minimum from the iterative results using three sets of family candidates: a set of Gaussian and Clayton in red ($\bB^1=\{G, C\}$), Gaussian only in green ($\bB^2=\{G\}$), and Clayton only in yellow ($\bB^3=\{C\}$). We also display below the iteration number, the selected family for $\bB^1$.
\begin{figure}
	\centering
	\includegraphics[width=0.95\textwidth]{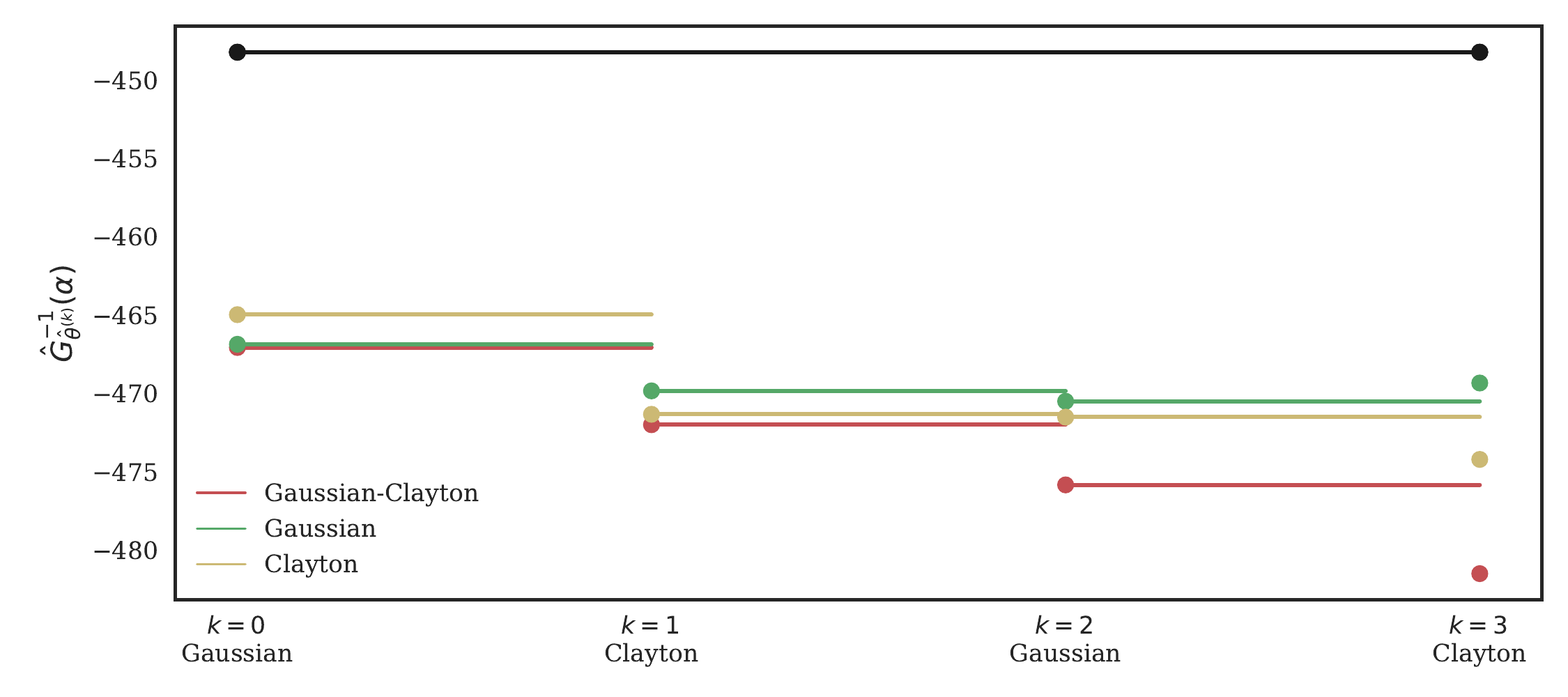}
	\caption{Quantile minimization for different set of family candidates $\bB$. The dark line shows the quantile at independence. The minimum at each iteration for the family candidates sets $\bB^1$, $\bB^2$ and $\bB^3$ respectively in red, green and yellow.}
	\label{fig:5:iterative_multi_k_3}
\end{figure}

At iteration $k=0$, the algorithm with the set $\bB^1$ has selected the Gaussian copula as the selected pair and the result is as expected equivalent as for the set $\bB^2$. At next iteration, a Clayton copula has been selected for algorithm with the set $\bB^1$, which slightly improves the minimization compared to the others. The improvement start at iteration $k=2$ where the Algorithm with the set $\bB^1$ minimizes more the output quantile than the other sets with only one copula family. At the last iteration, the algorithm with set $\bB^1$ selected a mix between Gaussian an Clayton families. This diversity seems to lead to better results than using only one family for every pairs. Testing multiple families is an interesting feature of the algorithm and is something that cannot be feasible for the grid-search approach. However, the cost for $\bB^1$ is twice larger than for the other methods.

\subsection{Industrial Application}

\subsubsection{Context}

We consider an industrial component belonging to a production unit. This component must maintain its integrity even in case of an accidental situation. For this reason, it is subject to a justification procedure by regulation authorities, in order to demonstrate its ability to withstand severe operating conditions. This undesirable event consists in the concomitance of three different factors:
\begin{itemize}
	\item the potential existence of small and undetectable manufacturing defects ;
	\item the exposition of the structure to an ageing phenomenon harming the material which progressively diminishes its mechanical resistance throughout its lifespan ;
	\item the occurrence of an accidental event generating severe constraints on the structure.
\end{itemize}
If combined, these three factors might lead to the initiation of a crack within the structure. Since no failure was observed until now, a structural reliability study should be conducted to check  the safety of the structure. To do so a thermal-mechanical code $\eta: \RR^d \rightarrow \RR^{+}$ was used, which calculates the ratio between the resistance and the stress acting on the component during a simulated accident. The numerical model depends on parameters affected by uncertainties quantified throughout numerous mechanical tests. Nevertheless, these experiments are mostly established individually and only few experiments involves simultaneously two parameters.
	
\subsubsection{Probabilistic model}

For this problem, we introduce $d=6$ random variables with predefined marginal distributions $(P_j)_{j=1\dots d}$. The dependence structure is however unknown. From the 15 pairs of variables, only the dependencies of two pairs are known: one is independent and the other follows a Gumbel copula with parameter 2.27. Therefore, we consider $p=13$ pairs of variables with unknown dependencies.

Given expert feedbacks, we restricted the exploration space $\bTheta$ by defining bounds for each pair of variables $(i, j) \in \Omega$ such that
$$
T_{c_{i, j}}(\tau_{i, j}^-) \leq \theta_{i, j} \leq T_{c_{i, j}}(\tau_{i, j}^+),
$$
where $\tau_{i, j}^-$ and $\tau_{i, j}^+$ are respectively the upper and lower kendall's correlation coefficient bounds for the dependence of the pair $(i, j)$ and $T_{c_{i, j}}$ is the transformation from Kendall's tau value to the copula parameter for the associated copula $c_{i, j}$. This choice enables to explore only realistic dependence structures. For these experiments we only considered Gaussian copulas.

\subsubsection{Results}

We consider the quantile at $\alpha=0.01$ as a quantity of interest. A first experiment is established with the incomplete probability structure: only the two pairs with a known dependence structure and all others at independence. Two other experiments are established: an exhaustive grid-search approach with a given vine structure and an iterative procedure with a maximum budget equivalent to the grid-search. A grid-size of $1000$ is chosen with $n=20,000$.

The results are displayed in Figure \ref{fig:5:result_application}. and has the same description as Figure \ref{fig:5:iterative_dim_10}. The quantile for the incomplete probability structure is approximately at 1.8. The grid-search and the iterative approaches found dependence structures leading to output quantile values close to 1.2 and 1.1 respectively. The minimum quantile from the iterative procedure is slightly lower than the grid-search approach. The problem dimension is not big enough to create make a significant difference between the methods. However, the resulting dependence structure from the iterative method is greatly simplified with only four pairs of variables, in addition to the already known pair.
\begin{figure}
	\centering
	\includegraphics[width=0.9\textwidth]{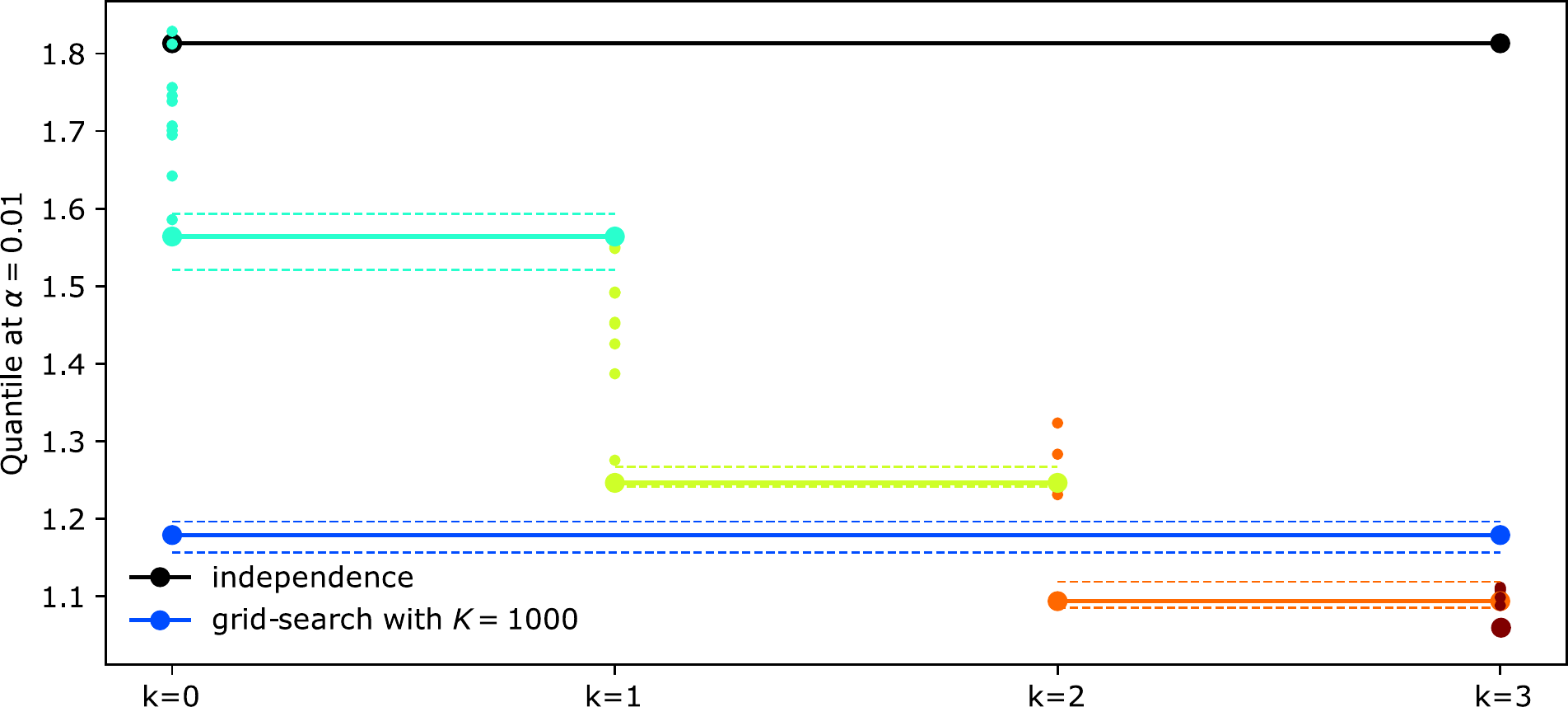}
	\caption{Minimization of the output quantile using a grid-search and the iterative procedure for $\alpha = 1\%$. The description is the same as in Figure \ref{fig:5:iterative_dim_10}.}
	\label{fig:5:result_application}
\end{figure}

This result highlights the risk of having an incomplete dependence structure in a reliability problem. In this application, the critical limit ({\it safety margin}) of the considered industrial component is 1. With the incomplete distribution of $\bX$, the output quantile is very high compared to the critical limit and states a high reliability of the component. Unfortunately, if we consider worst-case dependence structures, the output quantile is significantly minimized and becomes closer to the critical limit. Thus, if the true dependence structure is close to the obtained worst case dependence structure, the risk of over estimating the output quantile can be important.
\section{Conclusion and discussion}
\label{sec:conclusion}

%\subsection{A brief summary}

Incomplete information on inputs is an issue frequently encountered in structural reliability. Because safety analyses are mostly based on propagating random uncertainties through black-box computer models, the selection of a conservative dependence structure between input components appears as a requirement to define probabilistic worst cases. This article takes a first step towards such a methodology, by proposing a greedy, heuristic algorithm that explores a set of possible dependencies, taking advantage of the pair-copula construction (PCC) of multivariate probability distributions. Results of experiments conducted on toy and a real models illustrate the good behavior of the procedure: in situations where the monotonicity of the considered risk indicator (the output quantile) with respect to the inputs is postulated, a minimum value for the risk indicator is obtained using \FH bounds. Nonetheless, it is possible to exhibit situations where the algorithm detect other and more conservative dependence structures. This first step required a number of hypotheses and approximations that pave the way for future research. Besides, some perspectives arise from additional technical results. 

It would be interesting to improve the statistical estimation of the minimum quantile, given a dependence structure, by checking the hypotheses underlying the convergence results of Theorem \ref{theor:2:consistency_extremum_estimator}. Checking and relaxing these hypotheses should be conducted in relation with expert knowledge on the computer model $\eta$ and, possibly, a numerical exploration of its regularity properties. The grid search estimation strategy promoted in Section \ref{subsec:2:grid-search} arises from the lack of information about the convexity and the gradient of $\btheta\to \Gtheta ^{-1} (\alpha)$. However, the method remains basic and stochastic recursive algorithms, such as the Robbins-Munro algorithm \cite{robbins1951stochastic}, can be proposed and tested as possibly more powerful (faster) alternatives.

A significant issue, is the computational cost of the exploration of possible dependence structures. Reducing this cost while increasing the completeness of this exploration should be a main concern of future works. Guiding the exploration in the space of conditional bivariate copulas using enriching criteria and possible expert knowledge can facilitate the minimization. The Algorithm \ref{algo:4:iterative_algorithm} can also be improved using nonparametric bootstrap. This would quantify the estimation quality of the selected minimum quantile of each iteration. Note however that a seducing feature of an iterative procedure is the a priori possibility of its adaptation to situations where the computational model $\eta$ is time-consuming. In such cases, it is likely that Bayesian global optimization methods based on replacing the computer model by a surrogate model (e.g., a kriging-based meta-model) \cite{Osborne2009} should be explored, keeping in mind that nontrivial conservative correlations -- losses of quantile monotonicity-- can be due to edge effects (e.g., discontinuities) characterizing the computational model itself.

We noticed in our experiments on real case-studies that expert knowledge remains difficult to incorporate otherwise that using association and concordance measures, mainly since we are lacking of representation tools (e.g., visual) of the properties of multivariate laws that provide intelligible diagnostics. A first step towards the efficient incorporation of expert knowledge could be to automatize the visualization of the obtained vine structures, to simplify judgements about their realism.

Finally, another approach to consider could be to address the optimization problem \eqref{eq:2:general_problem} within the more general framework of optimal transport theory, and to take advantage of the many ongoing works in this area. Indeed, the problem \eqref{eq:2:general_problem} can be seen as a multi-marginal optimal transport problem (see \cite{refId0} for an overview). When $d = 2$, it corresponds respectively to the classical optimal transport problems of Monge and Kantorovich \cite{villani2008optimal}. However, the multimarginal theory is not as well understood as for the bimarginal case, and developing efficient algorithms for solving this problem remains also a challenging issue \cite{refId0}.

\section*{Acknowledgements}

The authors express their thanks to Gerard Biau (Sorbonne Universit\'e), Bertrand Iooss and Roman Sueur (EDF R{\&}D) for their useful comments and advice during the writing of this article.

\bibliographystyle{abbrv}
\bibliography{bibliography}

\begin{thebibliography}{}

\bibitem[Aas et~al., 2012]{aas2012truncated}
Aas, K., Czado, C., and Brechmann, E.~C. (2012).
\newblock Truncated regular vines in high dimensions with application to
  financial data.
\newblock {\em Canadian Journal of Statistics}, 40(1):68--85.

\bibitem[Aas et~al., 2009]{Aas2009mle}
Aas, K., Czado, C., Frigessi, A., and Bakken, H. (2009).
\newblock Pair-copula constructions of multiple dependence.
\newblock {\em Insurance, Mathematics and Economics}, 44:182--198.

\bibitem[Agrawal et~al., 2012]{Agrawal12}
Agrawal, S., Ding, Y., Saberi, A., and Ye, Y. (2012).
\newblock Price of correlations in stochastic optimization.
\newblock {\em Operations Research}, 60(1):150--162.

\bibitem[Bayarri et~al., 2007]{Bayarri2007}
Bayarri, M., Berger, J., Paulo, R., Sacks, J., Cafeo, J., Cavendish, J., Lin,
  C., and Tu, J. (2007).
\newblock A framework for validation of computer models.
\newblock {\em Technometrics}, 49:138--154.

\bibitem[Beaudoin and Lakhal-Chaieb, 2008]{beaudoin2008archimedean}
Beaudoin, D. and Lakhal-Chaieb, L. (2008).
\newblock Archimedean copula model selection under dependent truncation.
\newblock {\em Statistics in medicine}, 27(22):4440--4454.

\bibitem[Bedford and Cooke, 2001]{bedford2001probability}
Bedford, T. and Cooke, R.~M. (2001).
\newblock Probability density decomposition for conditionally dependent random
  variables modeled by vines.
\newblock {\em Annals of Mathematics and Artificial intelligence},
  32(1):245--268.

\bibitem[Bedford and Cooke, 2002]{bedford2002vines}
Bedford, T. and Cooke, R.~M. (2002).
\newblock Vines: A new graphical model for dependent random variables.
\newblock {\em Annals of Statistics}, 30(4):1031--1068.

\bibitem[Bedford et~al., 2006]{bedford2006}
Bedford, T., Quigley, J., and Walls, L. (2006).
\newblock Expert elicitation for reliable system design.
\newblock {\em Statistical Science}, 21(4):428--450.

\bibitem[Benoumechiara, 15  ]{depimpact}
Benoumechiara, N. (2015--).
\newblock {dep-impact}: Uncertainty quantification under incomplete probability
  information with {Python}.

\bibitem[Bobkov and Ledoux, 2014]{Bobkov14}
Bobkov, S. and Ledoux, M. (2014).
\newblock One-dimensional empirical measures, order statistics and kantorovich
  transport distances.
\newblock {\em preprint}.

\bibitem[Chazal et~al., 2015]{chazal2015rates}
Chazal, F., Massart, P., and Michel, B. (2015).
\newblock Rates of convergence for robust geometric inference.
\newblock {\em to appear in Electronic Journal of Statistics, arXiv preprint
  arXiv:1505.07602}.

\bibitem[Cherubini et~al., 2004]{Cherubini2004copula}
Cherubini, U., Luciano, E., and Vecchiato, W. (2004).
\newblock {\em Copula methods in finance}.
\newblock John Wiley \& Sons.

\bibitem[Council, 2012]{NRC2012}
Council, N.~R. (2012).
\newblock {\em Assessing the Reliability of Complex Models: Mathematical and
  statistical foundations of Verification, Validation and Uncertainty
  Quantification}.
\newblock Washington, DC: The National Academies Press.

\bibitem[Csorgo, 1983]{Csorgo83}
Csorgo, M. (1983).
\newblock {\em Quantile processes with statistical applications}, volume~42.
\newblock SIAM.

\bibitem[Czado, 2010]{czado2010pair}
Czado, C. (2010).
\newblock Pair-copula constructions of multivariate copulas.
\newblock In {\em Copula theory and its applications}, pages 93--109. Springer.

\bibitem[de~Rocquigny et~al., 2008]{de2008uncertainty}
de~Rocquigny, E., Devictor, N., and Tarantola, S. (2008).
\newblock {\em Uncertainty in industrial practice: a guide to quantitative
  uncertainty management}.
\newblock John Wiley \& Sons.

\bibitem[Demarta and McNeil, 2005]{demarta2005t}
Demarta, S. and McNeil, A.~J. (2005).
\newblock The t copula and related copulas.
\newblock {\em International statistical review}, 73(1):111--129.

\bibitem[Dissmann et~al., 2013]{dissmann2013selecting}
Dissmann, J., Brechmann, E.~C., Czado, C., and Kurowicka, D. (2013).
\newblock Selecting and estimating regular vine copulae and application to
  financial returns.
\newblock {\em Computational Statistics \& Data Analysis}, 59:52--69.

\bibitem[Dvoretzky et~al., 1956]{Dvoretzky56}
Dvoretzky, A., Kiefer, J., and Wolfowitz, J. (1956).
\newblock Asymptotic minimax character of the sample distribution function and
  of the classical multinomial estimator.
\newblock {\em The Annals of Mathematical Statistics}, pages 642--669.

\bibitem[Embrechts et~al., 2002]{embrechts2002correlation}
Embrechts, P., McNeil, A., and Straumann, D. (2002).
\newblock Correlation and dependence in risk management: properties and
  pitfalls.
\newblock {\em Risk management: value at risk and beyond}, pages 176--223.

\bibitem[Evans and Gariepy, 2015]{evans2015measure}
Evans, L.~C. and Gariepy, R.~F. (2015).
\newblock {\em Measure theory and fine properties of functions}.
\newblock CRC press.

\bibitem[Frechet, 1951]{Frechet51}
Frechet, M. (1951).
\newblock Sur les tableaux de correlation dont les marges sont dollnes.
\newblock {\em AIm. Unlv}.

\bibitem[Frees and Valdez, 1998]{frees1998understanding}
Frees, E.~W. and Valdez, E.~A. (1998).
\newblock Understanding relationships using copulas.
\newblock {\em North American actuarial journal}, 2(1):1--25.

\bibitem[Goda, 2010]{goda2010statistical}
Goda, K. (2010).
\newblock Statistical modeling of joint probability distribution using copula:
  Application to peak and permanent displacement seismic demands.
\newblock {\em Structural Safety}, 32(2):112--123.

\bibitem[Grigoriu and Turkstra, 1979]{Grigoriu79}
Grigoriu, M. and Turkstra, C. (1979).
\newblock Safety of structural systems with correlated resistances.
\newblock {\em Applied Mathematical Modelling}, 3(2):130--136.

\bibitem[Gruber and Czado, 2015]{Gruber2015bayesian}
Gruber, L. and Czado, C. (2015).
\newblock Sequential bayesian model selection of regular vine copulas.
\newblock {\em Bayesian Analysis}, 10(4):937--963.

\bibitem[Haff, rway]{Haff2016}
Haff, I.~H. (May 9-11, 2016, Oslo, Norway).
\newblock How to select a good vine.
\newblock {\em International FocuStat Workshop on Focused Information Criteria
  and Related Themes}.

\bibitem[Helton, 2011]{Helton2011}
Helton, J. (2011).
\newblock Quantification of margins and uncertainties: conceptual and
  computational basis.
\newblock {\em Reliability Engineering and System Safety}, 96:976--1013.

\bibitem[Hoeffding, 1940]{Hoeffding40}
Hoeffding, W. (1940).
\newblock Scale-invariant correlation theory.
\newblock {\em Schriften des Mathematischen Instituts und des Instituts fur
  Angewandte Mathematik der Universit at Berlin}, 5(3):181--233.

\bibitem[Iooss and Lema{\^\i}tre, 2015]{iooss2015review}
Iooss, B. and Lema{\^\i}tre, P. (2015).
\newblock A review on global sensitivity analysis methods.
\newblock In {\em Uncertainty Management in Simulation-Optimization of Complex
  Systems}, pages 101--122. Springer.

\bibitem[Jiang et~al., 2015]{jiang2015vine}
Jiang, C., Zhang, W., Han, X., Ni, B., and Song, L. (2015).
\newblock A vine-copula-based reliability analysis method for structures with
  multidimensional correlation.
\newblock {\em Journal of Mechanical Design}, 137(6):061405.

\bibitem[Joe, 1994]{Joe94}
Joe, H. (1994).
\newblock Multivariate extreme-value distributions with applications to
  environmental data.
\newblock {\em Canadian Journal of Statistics}, 22(1):47--64.

\bibitem[Joe, 1996]{joe1996families}
Joe, H. (1996).
\newblock Families of m-variate distributions with given margins and m (m-1)/2
  bivariate dependence parameters.
\newblock {\em Lecture Notes-Monograph Series}, pages 120--141.

\bibitem[Kendall, 1938]{kendall1938new}
Kendall, M.~G. (1938).
\newblock A new measure of rank correlation.
\newblock {\em Biometrika}, 30(1/2):81--93.

\bibitem[Kurowicka, 2011]{kurowicka2011estimation}
Kurowicka, D. (2011).
\newblock Optimal truncation of vines.
\newblock In Kurowicka, D. and Joe, H., editors, {\em Dependence Modeling: Vine
  Copula Handbook}. World Scientific Publishing Co.

\bibitem[Lemaire et~al., 2010]{lemaire2010}
Lemaire, M., Chateauneuf, A., and Mitteau, J.-C. (2010).
\newblock {\em Structural reliability}.
\newblock Wiley.

\bibitem[Malevergne et~al., 2003]{malevergne2003testing}
Malevergne, Y., Sornette, D., et~al. (2003).
\newblock Testing the gaussian copula hypothesis for financial assets
  dependences.
\newblock {\em Quantitative Finance}, 3(4):231--250.

\bibitem[Massart, 1990]{Massart90}
Massart, P. (1990).
\newblock The tight constant in the dvoretzky-kiefer-wolfowitz inequality.
\newblock {\em The Annals of Probability}, pages 1269--1283.

\bibitem[McKay et~al., 1979]{mckay1979comparison}
McKay, M.~D., Beckman, R.~J., and Conover, W.~J. (1979).
\newblock Comparison of three methods for selecting values of input variables
  in the analysis of output from a computer code.
\newblock {\em Technometrics}, 21(2):239--245.

\bibitem[McNeil and Ne{\v{s}}lehov{\'a}, 2009]{mcneil2009multivariate}
McNeil, A.~J. and Ne{\v{s}}lehov{\'a}, J. (2009).
\newblock Multivariate archimedean copulas, d-monotone functions and l-norm
  symmetric distributions.
\newblock {\em The Annals of Statistics}, pages 3059--3097.

\bibitem[Morales~N{\'a}poles, 2010]{morales2010bayesian}
Morales~N{\'a}poles, O. (2010).
\newblock {\em Bayesian belief nets and vines in aviation safety and other
  applications}.
\newblock PhD thesis, TU Delft, Delft University of Technology.

\bibitem[Morales-N{\'a}poles, 2011]{morales2011counting}
Morales-N{\'a}poles, O. (2011).
\newblock Counting vines.
\newblock {\em Dependence modeling: Vine copula handbook}, pages 189--218.

\bibitem[Nelsen, 2007]{Nelsen07}
Nelsen, R.~B. (2007).
\newblock {\em An introduction to copulas}.
\newblock Springer Science \& Business Media.

\bibitem[Newey and McFadden, 1994]{newey1994large}
Newey, W.~K. and McFadden, D. (1994).
\newblock Large sample estimation and hypothesis testing.
\newblock {\em Handbook of econometrics}, 4:2111--2245.

\bibitem[Osborne et~al., 2009]{Osborne2009}
Osborne, M., Garnett, R., and Roberts, S. (2009).
\newblock Gaussian processes for global optimization.
\newblock {\em Proceedings of the 3rd International Conference on Learning and
  Intelligent Optimization}.

\bibitem[{Pass, Brendan}, 2015]{refId0}
{Pass, Brendan} (2015).
\newblock Multi-marginal optimal transport: Theory and applications.
\newblock {\em ESAIM: M2AN}, 49(6):1771--1790.

\bibitem[Robbins and Monro, 1951]{robbins1951stochastic}
Robbins, H. and Monro, S. (1951).
\newblock A stochastic approximation method.
\newblock {\em The annals of mathematical statistics}, pages 400--407.

\bibitem[Saltelli et~al., 2000]{saltelli2000sensitivity}
Saltelli, A., Chan, K., Scott, E.~M., et~al. (2000).
\newblock {\em Sensitivity analysis}, volume~1.
\newblock Wiley New York.

\bibitem[Scarf et~al., 1958]{scarf1958min}
Scarf, H., Arrow, K., and Karlin, S. (1958).
\newblock A min-max solution of an inventory problem.
\newblock {\em Studies in the mathematical theory of inventory and production},
  10(2):201.

\bibitem[Schoelzel and Friederichs, 2008]{schoelzel2008multivariate}
Schoelzel, C. and Friederichs, P. (2008).
\newblock Multivariate non-normally distributed random variables in climate
  research--introduction to the copula approach.
\newblock {\em Nonlin. Processes Geophys.}, 15(5):761--772.

\bibitem[Sklar, 1959]{Sklar59}
Sklar, A. (1959).
\newblock {\em Fonctions de r{\'e}partition {\`a} n dimensions et leurs
  marges}, volume~8.
\newblock ISUP.

\bibitem[T. and T., 2003]{NIL03}
T., N. and T., A. (2003).
\newblock Models and model uncertainty in the context of risk analysis.
\newblock {\em Reliability Engineering and System Safety}, 79:309--317.

\bibitem[Tang et~al., 2013]{Tang2013impactcopula}
Tang, X.-S., Li, D.-Q., Rong, G., Phoon, K.-K., and Zhou, C.-B. (2013).
\newblock Impact of copula selection on geotechnical reliability under
  incomplete probability information.
\newblock {\em Computers and Geotechnics}, 49:264--278.

\bibitem[Tang et~al., 2015]{tang2015copula}
Tang, X.-S., Li, D.-Q., Zhou, C.-B., and Phoon, K.-K. (2015).
\newblock Copula-based approaches for evaluating slope reliability under
  incomplete probability information.
\newblock {\em Structural Safety}, 52:90--99.

\bibitem[Thoft-Christensen and S{\o}rensen, 1982]{Thoft82}
Thoft-Christensen, P. and S{\o}rensen, J.~D. (1982).
\newblock Reliability of structural systems with correlated elements.
\newblock {\em Applied Mathematical Modelling}, 6(3):171--178.

\bibitem[Van~der Vaart, 2000]{van2000asymptotic}
Van~der Vaart, A.~W. (2000).
\newblock {\em Asymptotic statistics}, volume~3.
\newblock Cambridge University Press.

\bibitem[Villani, 2008]{villani2008optimal}
Villani, C. (2008).
\newblock {\em Optimal transport: old and new}, volume 338.
\newblock Springer Science \& Business Media.

\end{thebibliography}

\appendix

\section{Proof of the consistency result}
\label{sec:proof_consitency}

The consistency of the estimator $\hat{\btheta}$ requires some regularity of the function $\btheta \mapsto \Gtheta$. This regularity can be also expressed in term of modulus of increase of the function $\btheta \mapsto \Qta$, on which some useful definitions and connections with the modulus of continuity are reminded.

\subsection{Modulus of increase of a cumulative distribution function}

Let us recall that a modulus of continuity is any real-extended valued function $\omega : [0, \infty) \mapsto [0, \infty) $ such that
$\lim_{\delta \rightarrow 0}  \omega(x) = \omega(0) = 0$.
The function $f : \RR \mapsto \RR $ admits  $\omega$ as modulus of continuity if for any $ (x,x') \in \RR^2$,
\begin{equation*}
	|f(x) -  f(x')| \leq \omega(|x - x'|).
\end{equation*}
Similarly, for some $x \in \RR$, the function $f$ admits $\omega$ as a {\it local} modulus of continuity if for any $ x' \in \RR^2$,
\begin{equation*}
	|f(x) -  f(x')| \leq \omega  (|x - x'|).
\end{equation*}
To control the deviation of the empirical quantile in the proof of Proposition \ref{prop:appx:uniform_convergence_DKn} further, we consider the modulus of continuity of the quantile functions $G^{-1}: [0, 1] \rightarrow \RR$ where $G$ is a distribution function on $\RR$. The quantile function being an increasing function, the {\it exact local modulus of continuity} of the quantile function $G^{-1}$ at $\alpha \in (0, 1)$ can be defined as
\begin{equation*}
	\label{eq:appx:local_modulus_continuity}
	\omega_{G^{-1}} (\epsilon, \alpha) := \max \left( G^{-1}(\alpha + \epsilon) - G^{-1}(\alpha), G^{-1}(\alpha) - G^{-1}(\alpha - \epsilon) \right).
\end{equation*}
In the proof of Proposition \ref{prop:appx:uniform_convergence_DKn}, we note that the continuity of a quantile function $G^{-1}$ can be connected to the increase of the distribution function $G$ (see also for instance Section A in~\cite{Bobkov14}).
Using the fact that the distribution function is increasing, we introduce the {\it local modulus of increase} of the distribution function $\epsilon_G$ at $y = G^{-1}(\alpha) \in \RR$ as:
\begin{equation*}
	\label{eq:appx:local_modulus_increase}
	\epsilon_G (\delta, y): = \min \left( G(y + \delta) - G(y), G(y) - G(y - \delta)\right).
\end{equation*}

\subsection{Proofs}

The estimator $\hat{\btheta}$ defined in \eqref{eq:2:estim} is an extremum-estimator (see for instance Section 2.1 of \cite{newey1994large}). The main ingredient to prove the consistency of this estimator is the uniform convergence in probability of the families of the  empirical quantiles $( \Qtha)_{\btheta \in \cD_{K_n}}  $ over  the family of  grids $ \cD_{K_n}$.

\begin{prop}
	\label{prop:appx:uniform_convergence_DKn}
	Let $\cD_{K_n}$ be defined as in Theorem~\ref{theor:2:consistency_extremum_estimator}. Let assume that \ref{hyp:2:continuous_distribution_D} and \ref{hyp:2:increasing_distribution_D} are both satisfied. Then, for all $\epsilon > 0$,
	$$
		\PP [ \sup_{\btheta \in \cD_{K_n}} | \Qtha - \Qta| > \epsilon ]\xrightarrow{n \rightarrow \infty} 0 .
	$$
\end{prop}
	
\begin{proof}[Proof of Proposition \ref{prop:appx:uniform_convergence_DKn}]
We first make the connection between the local continuity of the quantile function $\Gtheta^{-1}$ and the local increase of the distribution function $\Gtheta$. According to 
Assumption \ref{hyp:2:increasing_distribution_D}, we have that for any $\epsilon \in (0, \max((1-\alpha), \alpha))$, for any $\delta >0$ and for any $\btheta \in \cD$,
\begin{equation*}
	\label{eq:appx:equivalence_continuity_increase}
	(*)  : \left\{
	\begin{aligned}
		\Gtheta^{-1}(\alpha + \epsilon) - \Gtheta^{-1}(\alpha) < \delta \\
		\Gtheta^{-1} (\alpha) - \Gtheta^{-1}(\alpha - \epsilon) < \delta 
	\end{aligned}
	\right.
	\Longrightarrow 
	\left\{ 
	\begin{aligned}
		\Gtheta \left(\Gtheta^{-1}(\alpha + \varepsilon) \right)    < \Gtheta \left(\Gtheta^{-1}(\alpha) + \delta \right)    \\
		\Gtheta \left(\Gtheta^{-1}(\alpha) - \delta \right)    < \Gtheta \left(\Gtheta^{-1}(\alpha - \varepsilon ) \right)  
	\end{aligned}
	\right. .
\end{equation*}
Next, using basic properties of quantile functions (see for instance point {\bf ii} of Lemma 21.1 in \cite{van2000asymptotic} ) together with Assumption \ref{hyp:2:continuous_distribution_D}, we find that
$$
	\Gtheta \left(   \Gtheta^{-1}(\alpha + \epsilon) \right)  = \alpha  + \varepsilon = \Gtheta \left(   \Gtheta^{-1}(\alpha )\right)   + \epsilon
$$
and 
$$
	\Gtheta \left(   \Gtheta^{-1}(\alpha - \epsilon)\right)   = \alpha  - \varepsilon = \Gtheta \left(   \Gtheta^{-1}(\alpha )\right)   - \epsilon .
$$
Thus,
\begin{equation*}
	(*)   \Longrightarrow 
	\left\{ 
	\begin{aligned}
		\Gtheta \left( \Gtheta^{-1}(\alpha)  + \delta \right) -  \Gtheta \left( \Gtheta^{-1}(\alpha ) \right)   >  \epsilon \\
		\Gtheta \left( \Gtheta^{-1}(\alpha)   \right) -  \Gtheta \left( \Gtheta^{-1}(\alpha ) - \delta \right)   >  \epsilon
	\end{aligned}
	\right. .
\end{equation*}
We have shown that  any $\epsilon \in (0, \max((1-\alpha), \alpha))$, for any $\delta >0$ and for any $\btheta \in \cD$,
\begin{equation}
	\omega_{\Gtheta^{-1}}(\epsilon, \alpha)  > \delta 
	\Longrightarrow 
	 \epsilon_{\Gtheta} (\delta, \Gtheta^{-1}(\alpha)) < \epsilon.
	\label{eq:appx:equivalence_modulus_continuity_increase}
\end{equation}
We now prove the proposition. For any $n \geq 1$ and any $\epsilon > 0$, we have
\begin{align}
	\PP \left(\sup_{\btheta \in \cD_{K_n}} | \Qtha - \Qta | > \epsilon \right) 
	&= \PP \left(\bigcup_{\btheta \in \cD_{K_n}} \{| \Qtha - \Qta | > \epsilon \}\right) \nonumber\\
	&\leq \sum_{\btheta \in \cD_{K_n}} \Ptheta \left(| \Qtha - \Qta  | > \epsilon \right).
	\label{eq:appx:union_born_quantile}
\end{align}
Let $\xi_1, \dots,\xi_n$ be $n$ i.i.d. uniform random variables. The uniform empirical distribution function is defined by
$$
	\mathbb \UU (t) = \frac 1n \sum_{i=1}^n \mathbf{1}_{\xi_i \leq t} \quad \textrm{for }  0 \leq t \leq 1.
$$
The inverse uniform empirical distribution function is the function
$$
	\mathbb \UU^{-1}_n (u) =  \inf \{ t \, | \, \mathbb G _n (t) > u \} \quad \textrm{for }  0 \leq u \leq 1.
$$		
The empirical distribution function $\Gthetan$ can be rewritten as (see for instance \cite{van2000asymptotic}):
$$
	\Gthetan (y) \stackrel{\mathcal L}{=} \UU_n (\Gtheta (y))	 
$$
and as well for the quantile function,
\begin{equation*}
	\label{eq:appx:uniform_quantile_process}
	\Gthetan^{-1} (\alpha)  \stackrel{\mathcal L}{=}  \Gtheta^{-1} ( \UU_n^{-1}(\alpha)).
\end{equation*}
From Inequality \eqref{eq:appx:union_born_quantile}, we obtain
\begin{equation}
	\label{eq:appx:proba_bound_with_uniform_process}
	\sum_{\btheta  \in \cD_{K_n}} \Ptheta \left(| \Gthetan^{-1} (\alpha) - \Gtheta^{-1} (\alpha)  | > \epsilon \right) 
	= \sum_{\btheta  \in \cD_{K_n}} \Ptheta \left( | \Gtheta^{-1} ( \UU_n^{-1}(\alpha)) - \Gtheta^{-1}(\alpha) | > \epsilon \right)
\end{equation}
By definition of the local modulus of continuity $\omega_{\Gtheta^{-1}}$ of the quantile function $\Gtheta^{-1}$ at $\alpha$, we have
\begin{equation}
	\label{eq:appx:modulus_continuity_bound}
	| \Gtheta^{-1} ( \UU_n^{-1}(\alpha)) - \Gtheta^{-1}(\alpha) | 
	\leq
	\omega_{\Gtheta^{-1}} ( | \UU_n^{-1}(\alpha) - \alpha |, \alpha).
\end{equation}
Therefore, by replacing \eqref{eq:appx:modulus_continuity_bound} in \eqref{eq:appx:proba_bound_with_uniform_process} and using \eqref{eq:appx:equivalence_modulus_continuity_increase}, we obtain
\begin{align*}		
	\sum_{\btheta  \in \cD_{K_n}} \Ptheta \left( \left| \Gtheta^{-1} ( \UU_n^{-1}(\alpha)) - \Gtheta^{-1}(\alpha) \right| > \epsilon \right)
	&\leq 
	\sum_{\btheta  \in \cD_{K_n}} \Ptheta \left( \omega_{\Gtheta^{-1}} (|\UU_n^{-1}(\alpha) - \alpha|, \alpha) > \epsilon \right) \\
	&\leq
	\sum_{\btheta  \in \cD_{K_n}} \Ptheta \left( \epsilon_{\Gtheta}(\epsilon, \Gtheta^{-1} (\alpha) ) < | \UU_n^{-1} (\alpha) - \alpha| \right).
\end{align*}
Assumption \ref{hyp:2:increasing_distribution_D} then yields
\begin{equation}
	\label{eq:appx:proba_bound_min_modulus}
	\PP \left(\sup_{\btheta \in \cD_{K_n}} | \hat G_\theta ^{-1} (\alpha) - \Gtheta^{-1} (\alpha) | > \epsilon \right) 
	\leq
	K_n \PP \left( | \UU_n^{-1} (\alpha) - \alpha| > \underline{\epsilon}_{\cD} (\epsilon)  \right) .
\end{equation}
The DKW inequality \cite{Dvoretzky56} gives an upper bound of the probability of an uniform empirical process $\{|\UU_n(\alpha) - \alpha|\}$. As well for an uniform empirical quantile process $\{|\UU_n^{-1}(\alpha) - \alpha|\}$ (see for example Section 1.4.1 of \cite{Csorgo83}), such as $\forall \lambda > 0$:
$$
	\PP (\sup_{\alpha \in [0, 1]} |\UU_n^{-1}(\alpha) - \alpha|) \geq \lambda ) 
	\leq 
	C \exp (-2n \lambda^2) .
$$
Moreover, \cite{Massart90} proved that one can take $C=2$. Therefore, Equation \eqref{eq:appx:proba_bound_min_modulus} can be bounded using the DKW and 
$$
	\PP \left(\sup_{\btheta \in \cD_{K_n}} | \Qtha - \Gtheta^{-1} (\alpha) | > \epsilon \right) 
	\leq 
	2 K_n  \exp \left[-2n \underline{\epsilon}_{\cD}  ^2(\epsilon) \right]
	\xrightarrow{n \rightarrow \infty} 0 
$$
since $K_N \lesssim n^\beta$.
\end{proof}

A second requirement to get the consistency of the extremum estimator is the regularity of $\btheta \mapsto \Gtheta^{-1} (\alpha)$. This is shown in the next proposition.

\begin{prop}
	\label{prop:appx:continuity_quantile}
	Under Assumptions  \ref{hyp:2:continuity_copula}, \ref{hyp:2:continuous_distribution_D} and \ref{hyp:2:increasing_distribution_D}, the function
	\begin{align*}
		\cD &\longrightarrow \overline{Im (\eta)} \\
		\btheta &\longrightarrow \Gtheta^{-1} (\alpha)
	\end{align*}
	is continuous in $\btheta$ over $\cD$.
\end{prop}
			
\begin{proof}[Proof of Proposition \ref{prop:appx:continuity_quantile}]
According to Assumption \ref{hyp:2:continuity_copula}, for any $\btheta \in \bTheta$, the distribution  $\Ptheta$ admits a density function $\ftheta$ with respect to the Lebesgue measure on $\RR^d$ such that 
$$ 
	\ftheta (x_1, \dots,x_d) = \ctheta \left(  F_1(x_1, \dots, x_d ) \right)  f_1 (x_1 )\dots f_d (x_d), 
$$ 
where $f_j$ is the marginal density function of $X_j$, for $j=1, \dots, d$ and the Lebesgue measure on $\RR$. Moreover, for any $\bx \in \RR^d$, the function $\btheta \rightarrow \ftheta (x)$ is continuous in $\btheta$ over $\cD$. 

The domain $\cD \times [0, 1]^p$ is a compact set and according to Assumption \ref{hyp:2:continuity_copula}, there exists a constant $\bar c$ such that $\forall (\btheta, \bu) \in   \cD \times [0, 1]^d$, $\ctheta(\bu) \leq \bar c$. Consequently, we have
\begin{equation}
	\label{eq:appx:majoration_density}
	|\ftheta(x_1, \dots, x_d)| \leq \bar c \prod_{i=1}^d f_i (x_i) .
\end{equation}
For $\btheta \in \cD$ and for any $h >0$, we denote $y_h = \Gthetah^{-1} (\alpha)$. According to Assumption \ref{hyp:2:continuous_distribution_D} we have $\alpha = \Gthetah(y_h)$ and thus,
\begin{align}			
	\Gtheta^{-1} (\alpha) - \Gthetah^{-1} (\alpha) &= \Gtheta^{-1} (\alpha) - y_h  \notag \\
	&= 
	\Gtheta^{-1}(\Gthetah (y_h)) - \Gtheta^{-1} (\Gtheta(y_h))
	\label{eq:appx:GthetaGintheta}
\end{align}
Now, using Assumption \ref{hyp:2:increasing_distribution_D}, we have that $\Gtheta$ is strictly increasing in the neighborhood of $\Gtheta^{-1}(\alpha)$ and thus $\Gtheta^{-1}$ is continuous in the neighborhood of $\alpha$. Note that 
\begin{align*}
	|\Gthetah (y_h)  - \Gtheta (y_h) | 
	&= \left| \int_{\RR^d}     \mathbb 1_{\eta (\bx) \leq y_h} \mathrm d \Fthetah(\bx) - \int_{\RR^d}  \mathbb 1_{\eta (\bx) \leq y_h} \mathrm d F_{\theta}(\bx)   \right| \\
	&= \left|  \int_{\RR^d}  \left[   \fthetah (\bx) - \ftheta (\bx) \right]  \mathbb 1_{\eta (\bx) \leq y_h} \mathrm d \lambda(\bx) \right|  \\
	&\leq
	\int_{\RR^d} |\fthetah (\bx) - \ftheta (\bx)| \mathrm d \lambda(\bx) 	
\end{align*}
We then apply a standard dominated convergence theorem using \eqref{eq:appx:majoration_density} to get that
$$
\Gthetah(y_h) - \Gtheta (y_h) \xrightarrow{h \rightarrow 0} 0.
$$
This, with \eqref{eq:appx:GthetaGintheta} and with the continuity of $\btheta \mapsto \Gtheta$,  shows that
$$
	\Gtheta^{-1} (\alpha) - \Gthetah^{-1} (\alpha) \xrightarrow{h\rightarrow 0} 0.
$$
\end{proof}

We are now in position to prove Theorem~\ref{theor:2:consistency_extremum_estimator}.

\begin{proof}[Proof of Theorem \ref{theor:2:consistency_extremum_estimator}]
Under Assumptions \ref{hyp:2:continuous_distribution_D} and \ref{hyp:2:increasing_distribution_D}, Proposition~\ref{prop:appx:uniform_convergence_DKn} directly gives that for any $\varepsilon >0$,
$$
	P \left(   \left| \inf_{\theta \in \cD_{K_n}}  \Qtha -   \inf_{\theta \in \cD_{K_n}} \Qta  \right| > \varepsilon \right)  \xrightarrow{n \rightarrow \infty} 0
$$
which means that
\begin{equation}
	\label{eq:appx:cv1} 
	P \left(   \left| \widehat G_{\hat \btheta} ^{-1} (\alpha) -   \inf_{\btheta \in \cD_{K_n}} \Qta  \right| > \varepsilon \right)  \xrightarrow{n \rightarrow \infty} 0  .
\end{equation}
If Assumption~\ref{hyp:2:continuity_copula} is also satisfied, Proposition~\ref{prop:appx:continuity_quantile} together with \eqref{eq:2:approx_grid} give that $\inf_{\btheta \in \cD_{K_n}} \Qta$ tends to $\inf_{\theta \in \cD} \Qta $ as $n$ tends to infinity. Thus
\begin{equation}
	\label{eq:appx:cv2} 
	\inf_{\btheta \in \cD_{K_n}} G_{\theta}^{-1} (\alpha)  \xrightarrow{n \rightarrow \infty}  {G_C^{-1}}^\star (\alpha)  = G_{\theta^*_C}^{-1} (\alpha) 
\end{equation}
We then derive \eqref{eq:2:cv_quantile} from \eqref{eq:appx:cv1} and \eqref{eq:appx:cv2}.

We now assume that Assumption~\ref{hyp:2:min_unicity_D} is also satisfied. Let $\btheta^*$ be the unique minimizer of $\btheta \mapsto \Qta$.  Let $ h >0$ such that 
$ {B(\btheta^*, h )}^c :=  \{\btheta \in \cD  \: : \: \| \btheta - \btheta^\star \|_2\geq \varepsilon \} $ is not empty. According to Proposition \ref{prop:appx:continuity_quantile} and using the fact that $\cD$ is compact, we have
\begin{equation}
	\label{eq:appx:caracmin} 
	\sup_{ \btheta \in  {B(\theta^*, h )}^c}   |\Gtheta^{-1} (\alpha) - G_{\btheta^*}^{-1} (\alpha)| > 0 .
\end{equation}	
Consequently, for any $\forall h > 0$ small enough, there exists $\epsilon > 0$ such that
\begin{equation}
	\label{eq:appx:distance_from_minimum}
	|\Gtheta^{-1} (\alpha) - G_{\theta^*}^{-1} (\alpha)| \leq \epsilon \Longrightarrow | \btheta - \theta^*| < h
\end{equation}
Let $h > 0$ and take $\epsilon$ such that \eqref{eq:appx:distance_from_minimum} is satisfied for $h$. According to Proposition~\ref{prop:appx:uniform_convergence_DKn}, $\widehat G_{\hat \theta} ^{-1} (\alpha)  - G_{\hat \theta} ^{-1} (\alpha)  $ tends to zero in probability as $n$ tends to infinity. This, with \eqref{eq:2:cv_quantile}, shows that
$$
	P \left(   \left|G_{\hat \theta} ^{-1} (\alpha) -    G_{\theta^*}^{-1} (\alpha)  \right| > \varepsilon \right)  \xrightarrow{n \rightarrow \infty} 0 .
$$
We conclude using \eqref{eq:appx:distance_from_minimum}. 
\end{proof}

\section{Vine copulas}
\label{sec:vine_copulas}

\subsection{Definition}
\label{subsec:appx:definition_vines}

A vine model describes a $d$-dimensional pair-copula construction (PCC) and is a sequence of linked trees where the nodes and edges correspond to the $d(d-1)/2$ pair-copulas. According to Definition \ref{def:appx:r_vine} from \cite{bedford2001probability}, a vine structure is composed of $d-1$ trees $T_1, \dots, T_{d-1}$ with several conditions.

\begin{mydef}[R-vine]
	The sequence $\cV = (T_1, \dots, T_{d-1})$ is an R-vine on $n$ elements if 
	\begin{enumerate}
	\item $T_1$ is a tree with nodes $N_1 = \{ 1, \dots, d \}$ and a set of edges denoted $E_1$.
	\item For $i=2, \dots, d-1$, $T_i$ is a tree with nodes $N_i = E_{i-1}$ and edges set $E_i$.
	\item For $i = , \dots, d-1$ and $\{a, b \} \in E_i$ with $a = \{ a_1, a_2\}$ and $b = \{ b_1, b_2\}$ it must hold that $\#(a \cap b) = 1$ (proximity condition).
	\end{enumerate}
	\label{def:appx:r_vine}
\end{mydef}

Each tree $T_i$ is composed of $d-i+1$ nodes which are linked by $d-i$ edges for $i = 1, \dots, d-1$. A node in a tree $T_i$ must be an edge in the tree $T_{i-1}$, for $i=2, \dots, d-1$. Two nodes in a tree $T_i$ can be joined if their respective edges in tree $T_{i-1}$ share a common node, for $i=2, \dots, d-1$. The proximity condition, suggests that two nodes connected by an edge should share one variable from the conditioned set. The \textit{conditioning set} and \textit{conditioned set} are defined in Definition \ref{def:appx:conditioning_conditioned_sets} along with the \textit{complete union}. The complete union of an edge $e$ is a set of all unique variables contained in $e$.

\begin{mydef}[Complete union, conditioning and conditioned sets of an edge]
	Let $A_e$ be the complete union of an edge $e = \{a, b\} \in E_k$ in a tree $T_k$ of a regular vine $\cV$,
	$$
		A_e = \{ v \in N_1 | \exists e_i \in E_i, i = 1, \dots, k-1, \text{such that } v\in e_i \in \dots \in e_{k-1} \in e\}.
	$$
	The conditioning set associated with edge $e= \{ a, b\}$ is $D(e) := A_a \cap A_b$ and the conditioned sets associated with edge $e$ are $i(e) := A_a \backslash D(e)$ and $j(e) := A_b \backslash D(e)$. Here, $A \backslash B := A \cap B^c$ and $B^c$ is the complement of $B$.
	\label{def:appx:conditioning_conditioned_sets}
\end{mydef}

The conditioned and conditioning sets of an edge $e = \{a, b\}$ are respectively the symmetric difference and the intersection of the complete unions of $a$ and $b$. The conditioned and conditioning sets of all edges of $\cV$ are collected in a set called \textit{constraint set}. Each element of this set is composed of a pair of indices corresponding to the conditioned set and a set containing indices corresponding to the conditioning set, as shown in Definition \ref{def:appx:constraint_set}.
\begin{mydef}[Constraint set]
	The constrain set for $\cV$ is a set:
	$$
	\mathcal{C V} = \{ (\{ i(e), j(e) \}, D_e) | e \in E_i, e=\{ a, b\}, i=1, \dots, d-1\}
	$$
	\label{def:appx:constraint_set}
\end{mydef}
The pair-copula in the first tree characterize pairwise unconditional dependencies, while the pair-copula in higher order trees model the conditional dependency between two variables given a set of variables. The number of conditioning variables grows with the tree order. Note that a PCC where all trees have a path-like structure define the D-vine subclass while the star-like structures correspond to C-vine subclass. All other vine structures are called regular vines (R-vines) \cite{bedford2001probability}.

We illustrate the concept of a vine model with a $d=5$ dimensional example. For clarity reasons, we use the same simplifications as in Section \ref{subsec:4:estimation_fixed} which consider for instance $f_{1} = f_{1}(x_1)$, $f_{2} = f_{2}(x_2)$ and so on for higher order and conditioning. One possible PCC can be written for this 5-dimensional configuration:
\begin{align}
	f(x_1, x_2, x_3, x_4, x_5) 	&= f_1 \cdot f_2 \cdot f_3 \cdot f_4 \cdot f_5\ \text{\scriptsize{(margins)}} \nonumber\\
						\text{\scriptsize{(unconditional pairs)}}
						& \times
							c_{12} \cdot 
							c_{35} \cdot 
							c_{34} \cdot
							c_{24}
						 \nonumber \\
						\text{\scriptsize{(1st conditional pair)}} 
						& \times
							c_{14 | 3} \cdot 
							c_{23 | 4} \cdot 
							c_{45 | 3} 
						\nonumber \\
						\text{\scriptsize{(2nd conditional pair)}}
						& \times
							c_{15 | 34} \cdot
							c_{25 | 34} 
						\nonumber \\
						\text{\scriptsize{(3rd conditional pair)}}
						& \times
						c_{12 | 345} .
	\label{eq:appx:PCC_5d}
\end{align}
The vine structure associated to \eqref{eq:appx:PCC_5d} is illustrated in Figure \ref{fig:appx:R_vine_example_5d}. This graphical model considerably simplify the understanding and we observe that this model is a R-vine because there is no specific constraints on the trees.
\begin{figure}
	\centering
	\includegraphics[width=0.36\textwidth]{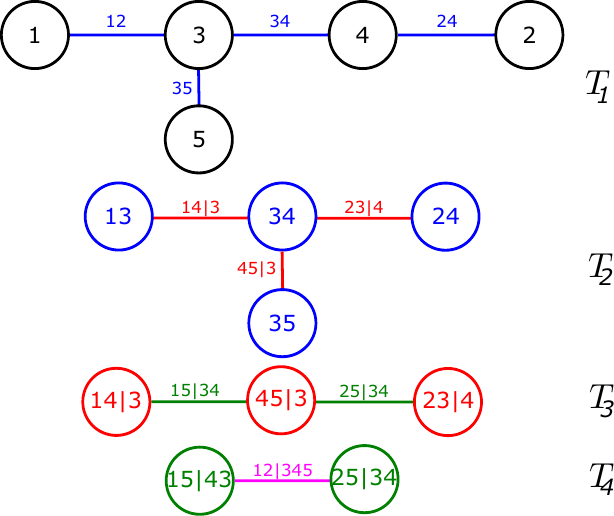}
	\caption{R-vine structure for $d=5$.}
	\label{fig:appx:R_vine_example_5d}
\end{figure}

A re-labeling of the variables can lead to a large number of different PCC. \cite{morales2011counting} calculated the number of possible vine structures with the dimension $d$ and shows that it becomes extremely large for high dimension problems. We illustrate below, using the same $d=5$ dimensional example, two other PCC densities:
\begin{minipage}{0.49\textwidth}
\begin{align}
	f_D &= f_1 \cdot f_2 \cdot f_3 \cdot f_4 \cdot f_5 \nonumber\\
						& \times
							c_{12} \cdot 
							c_{23} \cdot 
							c_{34} \cdot
							c_{45}
						 \nonumber \\
						& \times
							c_{13 | 2} \cdot 
							c_{24 | 3} \cdot 
							c_{35 | 4} 
						\nonumber \\
						& \times
							c_{14 | 23} \cdot
							c_{25 | 34} 
						\nonumber \\
						& \times
						c_{15 | 234} 
	\label{eq:appx:PCC_5d_d_vine}
\end{align}
\end{minipage}
\begin{minipage}{0.49\textwidth}
\begin{align}
	f_C	&= f_1 \cdot f_2 \cdot f_3 \cdot f_4 \cdot f_5 \nonumber\\
						& \times
							c_{12} \cdot 
							c_{13} \cdot 
							c_{14} \cdot
							c_{15}
						 \nonumber \\
						& \times
							c_{23 | 1} \cdot 
							c_{24 | 1} \cdot 
							c_{25 | 1} 
						\nonumber \\
						& \times
							c_{34 | 12} \cdot
							c_{35 | 12} 
						\nonumber \\
						& \times
						c_{45 | 123} 
	\label{eq:appx:PCC_5d_c_vine}
\end{align}
\vspace{0.1em}
\end{minipage}
where \eqref{eq:appx:PCC_5d_d_vine} and \eqref{eq:appx:PCC_5d_c_vine} respectively correspond to D-vine and C-vine structures and are represented in Figures \ref{fig:appx:d_vine_example_5d} and \ref{fig:appx:c_vine_example_5d}. As we can see in these examples, the D-vine have a constraint on each tree that gives a path-like arrangement of the nodes. The C-vine on the other hand only has one node connected to all others for each tree.
\begin{figure}
	\centering
	\subfloat[D-vine structure for $d=5$.]{
		\centering
		\includegraphics[width=.5\textwidth]{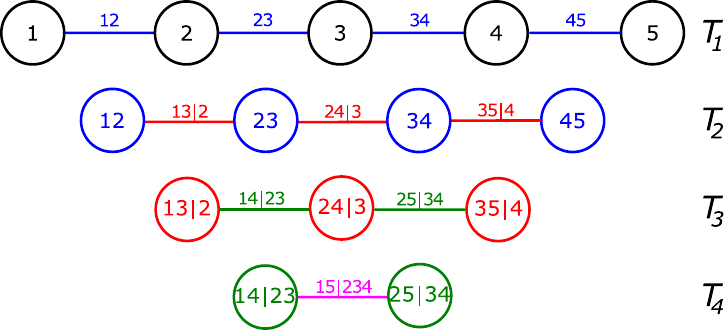}
		\label{fig:appx:d_vine_example_5d}
	}
	\subfloat[C-vine structure for $d=5$.]{
		\centering
		\includegraphics[width=0.5\textwidth]{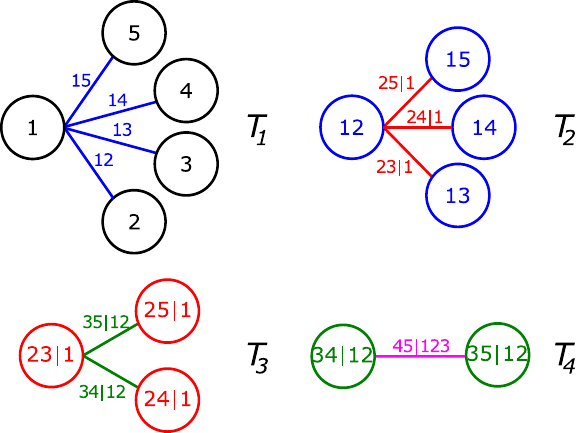}
		\label{fig:appx:c_vine_example_5d}
	}
	\caption{D-vine and C-vine structure for $d=5$.}
\end{figure}

An efficient way to store the information of a vine structure is proposed in \cite{morales2010bayesian} and is called a R-vine array. The approach uses the specification of a lower triangular matrix where the entries belong to $1, \dots, d$. Such matrix representation allows to directly derive the tree structure (or equivalently the associated PCC distribution). For more details, see \cite{morales2010bayesian}.

\subsection{Generating R-vine from an indexed list of pairs}
\label{subsec:appx:vine_construction}

The iterative procedure proposed in Section \ref{subsec:4:estimation_iterative}, described by Algorithm \ref{algo:4:iterative_algorithm}, minimizes the output quantile by iteratively determining the pairs of variables that influences the most the quantile minimization. At each iteration of the algorithm (step 1.a), a new vine structure is created by considering the list of influential pairs. The specificity of this vine creation is to consider the ranking of the list by placing the most influential pairs in the first trees of the R-vine. Thus, we describe in this section how to generate vine structure with the constraint of a given list of indexed pairs to fill in the structure. 

\subsubsection{The algorithm}

We consider the same notation as in Algorithm \ref{algo:4:iterative_algorithm}. Creating a vine structure from a given indexed list of pairs $\Omega_k$ is not straightforward. The difficulties come from respecting the ranking of $\Omega_k$ and the respect of the R-vine conditions. Indeed, the pairs cannot be append in the structure easily. The vine structure must respect these conditions, which can be sometime very restrictive. The procedure we proposed is detailed by the pseudo-code of Algorithm \ref{algo:appx:vine_construction} and can be greatly simplified in these few key steps:
\begin{itemize}
\item[1.] fill $\cV$ with the list $\Omega_k$,
\item[2.] fill $\cV$ with a permutation of $\Omega_{-k}$,
\item[3.] if $\cV$ is not a R-vine, then permute $\Omega_k$ and restart at step 1.
\end{itemize}
In step 1 and 2, the \textit{filling} procedure, detailed in Algorithm \ref{algo:appx:filling_vine_structure}, successively \textit{adds} the pairs of a list in the trees of a vine structure. Adding a pair $(i, j)$ in a tree $T_l$ associates $(i, j)$ with the conditioned set and determine a possible conditioning set $D$ from the previous tree such as a possible edge is $i, j | D$.

In step 2, because the ordering of $\Omega_{-k}$ is not important in the filling of $\cV$, the permutation of $\Omega_{-k}$ aims at finding a ranking such as $\cV$ leads to a R-vine.

In step 3, when the previous step did not succeeded and the resulting $\cV$ is not a R-vine structure, then the ranking of $\Omega_k$ is not possible and must be changed. The permutation of some elements of $\Omega_k$ must be done such as the ranking of the most influential pairs remains as close as possible to the initial one.

\SetKwFunction{FFill}{Fill}
\SetKwFunction{FCheck}{Check}
\SetKwFunction{FAdd}{Add}
\SetKwFunction{FFind}{FindConditioningSet}
\SetKwFunction{FFill}{Fill}

\begin{algorithm}[h]
	\caption{Generating a vine structure from a given list of indexed pairs $\Omega_k$}
	\label{algo:appx:vine_construction}

	\KwData{$\Omega_k$, $d$}
	\KwResult{A vine structure $\cV$.}

	$\Omega_k^{init} = \Omega_k$\;
	$k = 1$\;
	\Do{$\cV$ is not a R-vine}{
		\tcc{initialize $\cV$ with a first empty tree}
		$N_1 = (1, \dots, d)$\;
		$E_1 = ()$\;
		$\cV = ((N_1, E_1))$\;

		\tcc{filling $\cV$ with the list of selected pairs $\Omega_k$}
		$\cV = \FFill(\cV,\ \Omega_k,\ d)$\tcp*{See Algorithm \ref{algo:appx:filling_vine_structure}}

		\tcc{determining a permutation of $\Omega_{-k}$ that fills $\cV$}
		\For{$\Omega_{-k}^\pi \in \pi(\Omega_{-k})$}{
			\tcc{filling $\cV$ with the candidate pairs $\Omega_{-k}^\pi$}
			$\cV_\pi = \FFill(\cV,\ \Omega_{-k}^\pi,\ d)$\tcp*{See Algorithm \ref{algo:appx:filling_vine_structure}}
			\If{$\cV_\pi$ is a R-vine}{
				\tcc{a permutation worked $\rightarrow$ we quit the loop}
				\textbf{break}
			}
		}

		$\cV = \cV_\pi$\;
		\If{$\cV$ is not a R-vine}{
			\tcc{filling did not work $\rightarrow$ permute initial list $\Omega_k^{init}$}
			Get $\Omega_k$ by inverting pairs of $(\Omega_k^{init}$\;
			$k = k + 1$\;
		}
	}
\end{algorithm}

\subsubsection{Example}

For illustration, let's create a $d=5$ dimensional vine structure with the given list of pairs $\Omega_k = ((1, 2), (1, 3), (2, 3), (4, 5), (2, 4), (1, 5))$ using Algorithm \ref{algo:appx:vine_construction}. Using the original list $\Omega_k$, the $\FFill$ function may fail at line 7 of Algorithm \ref{algo:appx:vine_construction}, and more precisely, at line 15 of Algorithm \ref{algo:appx:filling_vine_structure}. Indeed, the first tree of $\cV$ does not validate the R-vine conditions. The tree is illustrated in Figure \ref{fig:appx:unique_graph} and as we can see, the nodes are not all connected into one single tree. Therefore, we permuted the list $\Omega_k$ by exchanging the pairs $(2, 4)$ and $(4, 5)$, as shown in Figure \ref{fig:appx:permutation_list_example}. This permutation now leads to a vine structure that respects the new ranked list $\Omega_k = ((1, 2), (1, 3), (2, 3), (2, 4), (4, 5), (1, 5))$ .
\begin{figure}
	\centering
	\includegraphics[width=0.4\textwidth]{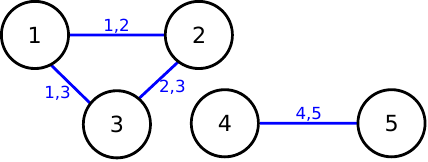}
	\caption{Example: first tree of a non valid vine structure for $d=5$ that does lead to a single connected tree.}
	\label{fig:appx:unique_graph}
\end{figure}
\begin{figure}
	\centering
	\includegraphics[width=0.4\textwidth]{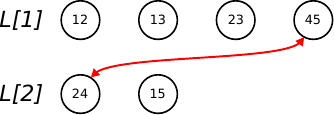}
	\caption{Example: exchange of elements of $\Omega_k$ in order to lead to a valid vine structure.}
	\label{fig:appx:permutation_list_example}
\end{figure}

\begin{algorithm}
	\caption{Filling a vine structure with a given list}
	\label{algo:appx:filling_vine_structure}

	\Fn{\FFill{$\cV$, $\Omega_k$, $d$}}{
		\tcc{
			$\cV$: an incomplete vine structure ,\\
			$\Omega_k$:, a list of indexed pairs\\
			$d$: the input dimension.
		}
		$l = |\cV|$\tcp*{number of existing trees}
		$(T_1, \dots, T_l) = \cV$\;
		$k = |T_l|$\tcp*{number of existing nodes in last tree}

		\tcc{loop over the list of pairs}
		\For{$(i, j) \in \Omega_k$}{
			$D = \emptyset$\;
			\If{$l >= 2$}{
				\tcc{conditioning set is only computed from $T_2$}
				$D = \FFind((i, j),\ N_{l-1})$\tcp*{See Algorithm \ref{algo:appx:find_conditioning_set}}
				\If {$D = \emptyset$}{
					\tcc{no conditioning set found $\rightarrow$ not possible}
					\KwRet False
				}
			}
			$E_l = E_l \cup i, j | D$\tcp*{add new edge in $E_l$}
			$T_l = (N_l, E_l)$\tcp*{update current tree}
			$\cV = (T_1, \dots, T_l)$\;

			\If{$k \geq d-l$}{
			\tcc{if tree $T_l$ is complete}
				\If{$\cV$ does not fulfill the R-vine conditions}{
					\tcc{the vine structure $\cV$ is not valid}
					\KwRet False
				}
				$k = 1$\;
				$l = l + 1$\;
				$N_l = E_{l-1}$\tcp*{nodes of next tree are the edges of previous tree}
			}
			\Else{
				$k = k + 1$\;
			}
		}

	\KwRet $\cV$
	}

\end{algorithm}

\begin{algorithm}
	\caption{Gets the conditioning set of a given conditioned set}
	\label{algo:appx:find_conditioning_set}
	\Fn{\FFind{$(i, j)$, $N_{-}$}}{
		\tcc{
			$(i, j)$: the conditioned set,\\
			$N_{-}$: list of nodes from the previous tree.
			}
		$D = \emptyset$\;
		\For{$a, b \in N_{-}$, with $a \neq b$}{
			\If{$i \in a$ and $j \in b$}{
				\If{$j \notin A_a$ and $i \notin A_b$}{
					\tcc{See Definition \ref{def:appx:conditioning_conditioned_sets}}
					$D = A_a \cap A_b$\;
					\textbf{break}\;
				}
			}
		}
		\KwRet $D$
	}
\end{algorithm}

\end{document}